\documentclass[acmsmall,nonacm]{acmart} 
\pdfoutput=1 

\usepackage{balance}
\usepackage{enumerate}
\usepackage{thm-restate}
\usepackage{hyperref}
\usepackage{cleveref}
\usepackage{stmaryrd}

\usepackage{amssymb}
\usepackage{amsmath}
\usepackage{tikz}
\usetikzlibrary{backgrounds}
\usetikzlibrary{arrows}
\usetikzlibrary{shapes,shapes.geometric,shapes.misc}

\tikzstyle{tikzfig}=[baseline=-0.25em,scale=0.5]

\pgfkeys{/tikz/tikzit fill/.initial=0}
\pgfkeys{/tikz/tikzit draw/.initial=0}
\pgfkeys{/tikz/tikzit shape/.initial=0}
\pgfkeys{/tikz/tikzit category/.initial=0}

\pgfdeclarelayer{edgelayer}
\pgfdeclarelayer{nodelayer}
\pgfsetlayers{background,edgelayer,nodelayer,main}

\tikzstyle{none}=[inner sep=0mm]

\newcommand{\tikzfig}[1]{%
{\tikzstyle{every picture}=[tikzfig]
\IfFileExists{#1.tikz}
  {\input{#1.tikz}}
  {%
    \IfFileExists{./figures/#1.tikz}
      {\input{./figures/#1.tikz}}
      {\tikz[baseline=-0.5em]{\node[draw=red,font=\color{red},fill=red!10!white] {\textit{#1}};}}%
  }}%
}

\tikzstyle{every loop}=[]
\usetikzlibrary{decorations.pathreplacing}
\usetikzlibrary{positioning}
\usetikzlibrary{patterns.meta}
\usepackage{tikz-cd}
\usepackage{setspace}
\usepackage[bibliography=common]{apxproof} 
\usepackage{bbding} 

\NewCommandCopy{\proofqedsymbol}{\qedsymbol}
\newcommand{\exampleqedsymbol}{{$\triangle$}}
\AtBeginEnvironment{proof}{\renewcommand{\qedsymbol}{\proofqedsymbol}}
\AtBeginEnvironment{example}{%
	\pushQED{\qed}\renewcommand{\qedsymbol}{\exampleqedsymbol}%
}
\AtEndEnvironment{example}{\popQED\endexample}
\AtBeginEnvironment{exa}{%
	\pushQED{\qed}\renewcommand{\qedsymbol}{\exampleqedsymbol}%
}
\AtEndEnvironment{exa}{\popQED\endexample}

\usepackage[xcolor, cleveref, notion, quotation, electronic]{knowledge}
\usepackage{mathcommand}
\knowledgeconfigure{protect quotation={tikzcd}}
\knowledgeconfigure{diagnose line=true, diagnose bar=true}

\definecolor{Dark Ruby Red}{HTML}{580507}
\definecolor{Dark Blue Sapphire}{HTML}{053641}
\definecolor{Dark Gamboge}{HTML}{be7c00}

\IfKnowledgePaperModeTF{
}{
    \knowledgestyle{intro notion}{color={Dark Ruby Red}, emphasize}
    \knowledgestyle{notion}{color={Dark Blue Sapphire}}
    \hypersetup{
        colorlinks=true,
        breaklinks=true,
        linkcolor={cGreen}, 
        citecolor={cGreen}, 
        filecolor={cGreen}, 
        urlcolor={cGreen},
    }
    \IfKnowledgeElectronicModeTF{
    }{
        \knowledgeconfigure{anchor point color={Dark Ruby Red}, anchor point shape=corner}
        \knowledgestyle{intro unknown}{color={Dark Gamboge}, emphasize}
        \knowledgestyle{intro unknown cont}{color={Dark Gamboge}, emphasize}
        \knowledgestyle{kl unknown}{color={Dark Gamboge}}
        \knowledgestyle{kl unknown cont}{color={Dark Gamboge}}
    }
} 

\knowledge{text={i.e.}, italic}
  | ie

\knowledge{text={s.t.}, italic}
  | st

\knowledge{text={e.g.}, italic}
  | eg

\knowledge{text={vs.}, italic}
  | vs

\knowledge{text={w.r.t.}, italic}
  | wrt

\knowledge{text={a.k.a.}, italic}
  | aka

\knowledge{text={w.l.o.g.}, italic}
  | wlog

\knowledge{text={cf.}, italic}
  | cf

\knowledge{text={iff}, italic}
  | iff
  
\knowledge{url={https://ctan.org/pkg/knowledge}, typewriter}
  | knowledge
  
\knowledge{url={https://github.com/remimorvan/Knowledge-Clustering}, typewriter}
| knowledge-clustering
 \knowledge{wrap=\textsf}
  | NL

\knowledge{notion, text={paraNL}, wrap=\textsf}
  | para-NL
  | paraNL

\knowledge{notion, wrap=\textsf}
  | FPT
  | fixed-parameter tractable

\knowledge{text={ExpSpace}, wrap=\textsf}
  | EXPSPACE
  | ExpSpace

\knowledge{text={2ExpSpace}, wrap=\textsf}
  | 2EXPSPACE
  | 2ExpSpace

\knowledge{text={PSpace}, wrap=\textsf}
  | PSpace
  | PSPACE

\knowledge{text={FP}, wrap=\textsf}
  | FP

\knowledge{text={NP}, wrap=\textsf}
  | NP

  \knowledge{text={P}, wrap=\textsf}
  | P

  \knowledge{text={\ensuremath{\# \textsf{P}}}, wrap=\textsf}
  | shP

  \knowledge{text={\ensuremath{\textsf{P}^{\# \textsf{P}}}}, wrap=\textsf}
  | PshP
  
\knowledge{text={\ensuremath{\textsf{FP}^{\# \textsf{P}}}}, wrap=\textsf}
  | FPshP

\knowledge{text={\ensuremath{\Pi^p_2}}, wrap=\textsf}
  | PiP2
  | Pi2
  | pi2
  | Pitwo
  | PiTwo
  | pitwo

\knowledge{url={https://en.wikipedia.org/wiki/Savitch\%27s_theorem}}
  | Savitch's Theorem

\knowledge{notion}
 | notion@notice
 | definition@notice

\knowledge{notion}
 | atom
 | atoms
 | (relational) atom

\knowledge{notion}
 | path atom
 | path atoms

 \knowledge{notion}
 | conn.
 | connected
 | connected query
 | connected queries

\knowledge{notion}
  | CQ
  | conjunctive query
  | CQs
  | conjunctive queries
  | Conjunctive Queries

\knowledge{notion}
 | regular path query
 | regular path queries
 | RPQs
 | RPQ

\knowledge{notion}
 | sjf-CQ
 | sjf-CQs

\knowledge{notion}
 | database
 | databases

\knowledge{notion}
 | exogenous

\knowledge{notion}
 | endogenous

\knowledge{notion}
 | relational schema
 | (relational) schema
 | schema

 \knowledge{notion}
 | relation name

\knowledge{notion}
 | fact
 | facts

\knowledge{notion}
 | active domain

\knowledge{notion}
 | graph database
 | graph databases
 | graphs

\knowledge{notion}
 | homomorphism
 | homomorphisms
 | homomorphic
 | homomorphically

\knowledge{notion}
 | $C$-homomorphism
 | $\aC$-homomorphism

 \knowledge{notion}
 | Boolean query
 | Boolean queries
 | Boolean

\knowledge{notion}
 | closed under homomorphisms
 | hom-closed

\knowledge{notion}
 | closure under $C$-homomorphisms
 | $C$-hom-closed
 | $\aC$-hom-closed
 | $C_2$-hom-closed

\knowledge{notion}
 | constants
 | constant

\knowledge{notion}
 | variables
 | variable

\knowledge{notion}
 | self-join-free

\knowledge{notion}
 | support

\knowledge{notion}
 | minimal support
 | minimal supports

\knowledge{notion}
 | Shapley value
 | Shapley values

\knowledge{notion}
 | cooperative game
 | cooperative games
 | games
 | game

 \knowledge{notion}
 | $(\D ,q)$-game

\knowledge{notion}
 | wealth function

\knowledge{notion}
 | tuple-independent probabilistic database
 | probabilistic database

\knowledge{notion}
 | probabilistic query evaluation

\knowledge{notion}
 | single probability query evaluation

\knowledge{notion}
 | single proper probability query evaluation
 
\knowledge{notion}
 | associate@pdb
 | associated partitioned database
 | associated@pdb

\knowledge{notion}
 | support-independent

\knowledge{notion}
 | pseudo-connected
 | pseudo-connectedness

\knowledge{notion}
 | variable-connected

\knowledge{notion}
 | leak

\knowledge{notion}
 | decomposition
 | decomposable
 | decomposability
 | decomposable into

\knowledge{notion}
 | polynomial-time reductions
 | polynomial-time Turing-reductions
 | polynomial-time Turing-reduction
 | polynomial-time reduction

\knowledge{notion}
 | unions of conjunctive queries
 | unions of Conjunctive Queries
 | union of conjunctive queries
 | UCQs
 | UCQ

\knowledge{notion}
 | CRPQs
 | CRPQ
 | conjunctive regular path query

\knowledge{notion}
 | unions of conjunctive regular path queries
 | UCRPQ
 | UCRPQs

\knowledge{notion}
 | partitioned databases
 | partitioned database

\knowledge{notion}
 | generalized supports for $q$ in $\D$
 | generalized support for $q$ in
 | generalized support for
 | generalized support
 | generalized supports
 | generalized supports for
 | generalized supports for 

\knowledge{notion}
 | terms
 | term

\knowledge{notion}
 | Infinitary unions of CQs
 | UCQ$^\infty$
 
\knowledge{notion}
 | DNF formula
 | DNF

 \knowledge{notion}
 | duplicable singleton support
 | dss

 \knowledge{notion}
 | pseudo-connected with an unshared constant
 | pseudo-connected queries with an unshared constant

 \knowledge{notion}
 | sjf-CRPQ

\knowledge{notion}
 | cc-disjoint-CRPQ

\knowledge{notion}
 | unshared constants
 | unshared constant

\knowledge{notion}
 | decomposable with an unshared constant 
 | decomposable query with an unshared constant

\knowledge{notion}
 | hierarchical
 | non-hierarchical

\knowledge{notion}
 | graph query
 | graph queries

\knowledge{notion}
 | safe

 \knowledge{notion}
 | unbounded
 
 \knowledge{notion}
 | relevant to
 | irrelevant to
 
 \knowledge{notion,text={\ensuremath{\text{sjf-CQ}^{\textup{g}\lnot}}}}
 | sjfCQgneg

 \knowledge{notion}
 | self-join-free conjunctive queries with safe negations

 \knowledge{notion}
 | incidence graph

\knowledge{notion}
 | satisfies

\knowledge{notion}
 | generalized model counting problem

\knowledge{notion}
 | fixed-size GMC problem
 | fixed-size generalized model counting problem

\knowledge{notion}
 | model counting problem

\knowledge{notion}
 | fixed-size model counting problem

 \knowledge{notion}
 | island
 | island support
 
 \knowledge{notion}
 | \Ai
 | \Ain
 | \Aix 
\usepackage{dutchcal} 
\usepackage{xspace} 
\usepackage{booktabs} 
\usepackage{bbding} 

\usepackage{caption}
\usepackage{subcaption}

\definecolor{Desire}{HTML}{eb3b5a} 
\definecolor{Boyzone}{HTML}{2d98da} 
\definecolor{Royal Blue}{HTML}{3867d6} 
\definecolor{NYC Taxi}{HTML}{f7b731} 
\definecolor{Beniukon Orange}{HTML}{fa8231}
\definecolor{Algal Fuel}{HTML}{20bf6b} 
\definecolor{Innuendo}{HTML}{a5b1c2} 
\definecolor{Twinkle Blue}{HTML}{d1d8e0} 
\definecolor{Blue Horizon}{HTML}{4b6584} 
\definecolor{Gloomy Purple}{HTML}{8854d0} 

\colorlet{cBlue}{Royal Blue}
\colorlet{cYellow}{NYC Taxi}
\colorlet{cOrange}{Beniukon Orange}
\colorlet{cGreen}{Algal Fuel}
\colorlet{cRed}{Desire}
\colorlet{cGrey}{Innuendo}
\colorlet{cDarkGrey}{Blue Horizon}
\colorlet{cLightGrey}{Twinkle Blue}
\colorlet{cPurple}{Gloomy Purple}

\renewcommand{\epsilon}{\varepsilon}

\newif\ifproofappendix

\newrobustcmd\introinrestatable[1]{%
\ifproofappendix%
\kl{#1}%
\else%
\intro{#1}%
\fi%
}

\newrobustcmd\introinrestatableopt[1]{%
\ifproofappendix%
\kl[#1]{#1}%
\else%
\intro[#1]{#1}%
\fi%
}

\newrobustcmd\recall[1]{
  \proofappendixtrue%
    #1*
  \proofappendixfalse%
}

\usepackage[backgroundcolor=orange!20, textcolor={Dark Ruby Red}, textsize=tiny, 
]{todonotes}
\definecolor{green}{RGB}{0,120,0}
\definecolor{hlyellow}{RGB}{250, 250, 190}
\definecolor{diegoeditcolor}{RGB}{210,210,255}
\definecolor{meghyneditcolor}{RGB}{210,255,210}
\definecolor{pierreeditcolor}{RGB}{255, 225, 186}

\definecolor{light-gray}{gray}{0.9}
\newcommand{\changed}[1]{#1}
\newcommand{\siderev}[1]{}
\newcommand{\rev}[1]{}

\renewcommand{\phi}{\varphi}
\renewcommand{\leq}{\leqslant}
\renewcommand{\geq}{\geqslant}
\renewcommand{\emptyset}{\varnothing}
\newcommand{\partsof}[1]{\wp(#1)}
\newcommand{\IQ}{\mathbb{Q}} 
\newcommand{\IN}{\mathbb{N}} 
\knowledgenewrobustcmd{\dcup}{\mathop{\cmdkl{\uplus}}} 

\newcommand{\set}[1]{\{#1\}}
\newrobustcmd{\defeq}{\mathrel{\hat{=}}}
\newcommand{\+}[1]{\mathcal{#1}}
\newrobustcmd{\Nat}{\mathbb{N}}
\newrobustcmd\pset[1]{\+P(#1)} 
\newrobustcmd{\pto}{\rightharpoonup} 
\newrobustcmd{\poly}{\mathop{\textrm{poly}}} 

\knowledgenewrobustcmd{\polyrx}{\mathrel{\cmdkl{\le_{\textit{poly}}}}} 
\knowledgenewrobustcmd{\polyeq}{\mathrel{\cmdkl{\equiv_{\textit{poly}}}}} 

\knowledgenewrobustcmd{\homto}{\mathrel{\cmdkl{\xrightarrow{\smash{\textit{hom}}}}}}
\knowledgenewrobustcmd{\Chomto}[1][C]{\mathrel{\cmdkl{\xrightarrow{\smash{{#1}\textit{-hom}}}}}}

\newcommand{\A}{\+{A}}
\newcommand{\W}{\+{W}}

\knowledgenewrobustcmd{\Const}{\cmdkl{\textup{Const}}}
\knowledgenewrobustcmd{\Var}{\cmdkl{\textup{Var}}}
\newcommand{\subendo}{\textup{\textsf{n}}}
\newcommand{\subexo}{\textup{\textsf{x}}}
\newcommand{\D}{\+{D}} 
\knowledgenewrobustcmd{\Dn}[1][\D]{#1_{\cmdkl{\subendo}}}
\knowledgenewrobustcmd{\Dx}[1][\D]{#1_{\cmdkl{\subexo}}}

\knowledgenewrobustcmd{\atoms}{\cmdkl{\textit{atoms}}}
\knowledgenewrobustcmd{\vars}{\cmdkl{\textit{vars}}}
\knowledgenewrobustcmd{\const}{\cmdkl{\textit{const}}}
\knowledgenewrobustcmd{\mterms}{\cmdkl{\textit{term}}}
\knowledgenewrobustcmd{\CQneg}{\ensuremath{\cmdkl{\textup{CQ}^\lnot}}}
\knowledgenewrobustcmd{\RAneg}{\ensuremath{\cmdkl{\textup{1RA}^-}}}

\newrobustcmd\pitwo{\ensuremath{\Pi^p_2}}
\newrobustcmd\sigmatwo{\ensuremath{\Sigma^p_2}}

\knowledgenewrobustcmd{\Sh}{\cmdkl{\mathrm{Sh}}} 
\knowledgenewrobustcmd{\Sym}{\cmdkl{\mathbb{P}}} 
\knowledgenewrobustcmd{\scorefun}[1][q]{\cmdkl{\mathbf{v}_{#1}}} 
\knowledgenewrobustcmd{\Shapley}[1]{\cmdkl{\textup{SVC}_{#1}}} 
\knowledgenewrobustcmd{\maxShapley}[1]{\cmdkl{\textup{max-SVC}_{#1}}} 
\knowledgenewrobustcmd{\ShapleyConst}[1]{\cmdkl{\textup{SVC}^{\textit{const}}_{#1}}}
\knowledgenewrobustcmd{\ShapleynConst}[1]{\cmdkl{\textup{SVC}^{\textsf{\textup{n}},\textit{const}}_{#1}}}
\knowledgenewrobustcmd{\Shapleyn}[1]{\cmdkl{\textup{SVC}^{\textsf{\textup{n}}}_{#1}}} 

\knowledgenewrobustcmd{\GMC}[1]{\cmdkl{\textup{GMC}_{#1}}}
\knowledgenewrobustcmd{\FGMC}[1]{\cmdkl{\textup{FGMC}_{#1}}}
\knowledgenewrobustcmd{\FGMCConst}[1]{\cmdkl{\textup{FGMC}^{\textit{const}}_{#1}}}
\knowledgenewrobustcmd{\FMCConst}[1]{\cmdkl{\textup{FMC}^{\textit{const}}_{#1}}}
\knowledgenewrobustcmd{\MC}[1]{\cmdkl{\textup{MC}_{#1}}}
\knowledgenewrobustcmd{\FMC}[1]{\cmdkl{\textup{FMC}_{#1}}}

\knowledgenewrobustcmd{\Proba}{\cmdkl{\textup{Pr}}} 
\knowledgenewrobustcmd{\PQE}[1]{\cmdkl{\textup{PQE}_{#1}}} 
\knowledgenewrobustcmd{\PQEP}[2]{\cmdkl{\textup{PQE}_{#1}(#2)}} 
\knowledgenewrobustcmd{\SPQE}[1]{\cmdkl{\textup{SPQE}_{#1}}} 
\knowledgenewrobustcmd{\SPPQE}[1]{\cmdkl{\textup{SPPQE}_{#1}}} 
\knowledgenewrobustcmd{\PQEPhalf}[1]{\cmdkl{\textup{PQE}_{#1}\left(\frac 1 2\right)}}
\knowledgenewrobustcmd{\PQEPhalfOne}[1]{\cmdkl{\textup{PQE}_{#1}\left(\frac 1 2 ; 1\right)}}

\knowledgenewrobustcmd{\aC}{\cmdkl{C}} 
\knowledgenewrobustcmd{\Ai}[1][]{\cmdkl{\A^{i}_{#1}}} 
\newcommand{\Ain}{\Ai[\subendo]}
\newcommand{\Aix}{\Ai[\subexo]}

\knowledgenewrobustcmd{\incgraph}[1]{\cmdkl{\mathbf{G}_{#1}}} 

\graphicspath{{../fig/}}

\newtheorem{claim}{Claim}[section]
\newtheorem{remark}{Remark}[section]
\newtheoremrep{proposition}{\sc Proposition}[section]
\newtheoremrep{lemma}{\sc Lemma}[section]
\newtheoremrep{claim}{\sc Claim}[section]
\Crefname{claim}{claim}{claims}
\Crefname{claim}{Claim}{Claims}
\newtheoremrep{corollary}{\sc Corollary}[section]
\AtBeginDocument{%
  }

\begin{document}

\title{When is Shapley Value Computation a Matter of Counting?}

\author{Meghyn Bienvenu}
\email{meghyn.bienvenu@cnrs.fr}
\affiliation{%
  \institution{Univ. Bordeaux, CNRS, Bordeaux INP, LaBRI, UMR 5800}
  \city{F-33400, Talence}
  \country{France}
  \postcode{F-33400}
}
\orcid{0000-0001-6229-8103}
\author{Diego Figueira}
\email{diego.figueira@cnrs.fr}
\affiliation{%
  \institution{Univ. Bordeaux, CNRS, Bordeaux INP, LaBRI, UMR 5800}
  \city{F-33400, Talence}
  \country{France}
  \postcode{F-33400}
}
\orcid{0000-0003-0114-2257}
\author{Pierre Lafourcade}
\email{pierre.lafourcade@u-bordeaux.fr}
\orcid{0009-0004-4810-1289}
\affiliation{%
  \institution{Univ. Bordeaux, CNRS, Bordeaux INP, LaBRI, UMR 5800}
  \city{F-33400, Talence}
  \country{France}
  \postcode{F-33400}
}


\begin{abstract}

The Shapley value provides a natural means of quantifying the 
contributions of facts to database query answers.
In this work, we seek to broaden our understanding of
Shapley value computation (SVC) in the database setting by revealing how it
relates to Fixed-size Generalized Model Counting (FGMC), which is the problem of 
computing the number of sub-databases of a given size and containing a given set of assumed facts 
that satisfy a fixed query.
Our focus will be on explaining the difficulty of SVC via FGMC, and to this end,
we identify 
general conditions on queries which enable reductions from FGMC to SVC.
As a byproduct, 
we not only obtain alternative explanations for existing hardness results for SVC,
but also new complexity results. 
In particular, we establish "FP"-"shP" complexity dichotomies
for constant-free unions of connected CQs and connected homomorphism-closed graph queries. 
We also consider some variants of the SVC problem, by disallowing assumed facts or quantifying the contributions of constants rather than facts. 

 \end{abstract}



\received{December 2023}
\received[revised]{February 2024}
\received[accepted]{March 2024}

\maketitle

\noindent
\raisebox{-.4ex}{\HandRight}\ \ This pdf contains internal links: clicking on a "notion@@notice" leads to its \AP ""definition@@notice"".\footnote{This result was achieved by using the "knowledge" package and its companion tool "knowledge-clustering".}

\color{black}
\section{Introduction}
\label{sec:intro}


The Shapley value \cite{shapley:book1952} is a well-known measure for distributing wealth among players in cooperative games.
It enjoys many desirable properties
and has found application in numerous areas, including databases, where it provides a natural means of 
quantifying the contribution of a database element (typically a tuple) to a query result.
The problem of Shapley value computation ($\Shapley{}$) for database queries 
has received considerable attention lately \cite{livshitsShapleyValueTuples2021,khalilComplexityShapleyValue2023,deutchComputingShapleyValue2022a,karaShapleyValueModel2023,reshefImpactNegationComplexity2020}. 
It is known to be computationally challenging -- "shP"-hard and "FPshP"-complete in data complexity under polynomial-time reductions -- 
which  motivates the question of identifying those 
queries for which $\Shapley{}$ 
 is tractable ("ie", in "FP"). It has been conjectured \cite{deutchComputingShapleyValue2022a,karaShapleyValueModel2023} that a "FP"/"shP"-hard dichotomy holds for "unions of conjunctive queries", as is the case for query evaluation in tuple-independent probabilistic databases \cite{dalviDichotomyProbabilisticInference2012} and for the generalized model counting problem \cite{kenigDichotomyGeneralizedModel2021}.
At present, this conjecture has 
been only confirmed for the subclass of 
self-join-free conjunctive queries (sjf-CQs) \cite{livshitsShapleyValueTuples2021}. 
There has however been progress on classifying the complexity of $\Shapley{}$ for 
 other query classes,
with "FP"/"shP"-hard dichotomies proven for 
sjf-CQs with safe negations \cite{reshefImpactNegationComplexity2020} and regular path queries \cite{khalilComplexityShapleyValue2023}.

The dichotomy established in \cite{livshitsShapleyValueTuples2021} for $\Shapley{}$ for sjf-CQs
identifies precisely the same tractable queries (the hierarchical ones) as the dichotomies for probabilistic query evaluation ($\PQE{}$)
and generalized model counting ($\GMC{}$),
despite having been obtained through different means. 
This raised the intriguing question of whether the Shapley dichotomy was a consequence of these existing dichotomies,
and more generally,  what is the precise relationship holding between these problems. 
Subsequent work \cite{deutchComputingShapleyValue2022a} provided a partial answer by exhibiting a general
polynomial-time  reduction (for all queries, not just sjf-CQs) of $\Shapley{}$ to $\PQE{}$, thereby showing how the tractability
results for $\PQE{}$ can be transferred to $\Shapley{}$. 

\begin{figure}
	\centering
	\hspace{-.1\linewidth}
	\begin{subfigure}{0.51\linewidth}
		\resizebox{\linewidth}{!}{\definecolor{vlightgray}{HTML}{EBEBEB}
\definecolor{puregreen}{rgb}{.15, .55, 0}

\begin{tikzpicture}[yscale=1.1]
	\small
	\node [] (18) at (4.25, -2.75) {};
	\node [] (19) at (4.25, 3.25) {};
	\node [] (20) at (-3.2, 3.25) {};
	\node [] (22) at (-4, -2.75) {};
	\node [] (24) at (4.25, -1.25) {};
	\node [] (25) at (-3.2, -1.25) {};
	\node [] (27) at (4.25, -1) {};
	\node [] (29) at (-4, -1) {};
	\node [] (30) at (-3.2, 1) {};
	\node [] (33) at (-4, 1.25) {};
	\node [] (38) at (-3.2, 1.25) {};
	\node [] (39) at (-1.25, 2) {};
	\fill [color=lightgray] (33.center) rectangle (20.center);
	\fill [color=lightgray] (29.center) rectangle (30.center);
	\fill [color=lightgray] (22.center) rectangle (25.center);
	\fill [color=vlightgray] (30.center) rectangle (27.center);
	\fill [color=vlightgray] (25.center) rectangle (18.center);
	\fill [color=vlightgray] (38.center) rectangle (19.center);

	\node [draw, rectangle, fill=white] (0) at (-2.2, 0.5) {$\MC{}$};
	\node [draw, rectangle, fill=white] (1) at (0.5, 0.5) {$\GMC{}$};
	\node [draw, rectangle, color=red, fill=white] (3) at (-0.7, 1.75) {$\SPQE{}$};
	\node [draw, rectangle, color=red, fill=white] (4) at (1.75, 1.75) {$\SPPQE{}$};
	\node [draw, rectangle, fill=white] (5) at (-2.2, 2.75) {$\PQEPhalf{}$};
	\node [draw, rectangle, fill=white] (6) at (0.5, 2.75) {$\PQEPhalfOne{}$};
	\node [draw, rectangle, fill=white] (7) at (3, 1.75) {$\PQE{}$};
	\node [draw, rectangle, fill=white] (9) at (-0.7, -0.5) {$\FMC{}$};
	\node [draw, rectangle, fill=white] (10) at (1.75, -0.5) {$\FGMC{}$};
	\node [draw, rectangle, fill=white] (11) at (-0.7, -1.75) {$\Shapleyn{}$};
	\node [draw, rectangle, fill=white] (12) at (1.75, -1.75) {$\Shapley{}$};
	\node [align=center,text width = 3.5em] (13) at (3, 0.5) {\begin{spacing}{0.5}
			\scriptsize "UCQ" dichotomy\\ \scriptsize
			\cite{dalviDichotomyProbabilisticInference2012,kenigDichotomyGeneralizedModel2021}
	\end{spacing}};
	\node [] (14) at (3, 2.75) {\scriptsize sjf-$\RAneg$ dichotomy \cite{finkDichotomiesQueriesNegation2016}};
	\node [] (15) at (1.75, -2.55) {\scriptsize "sjf-CQ" dichotomy \cite{livshitsShapleyValueTuples2021}};
	\node [] (16) at (3, -2.25) {\scriptsize "RPQ" dichotomy \cite{khalilComplexityShapleyValue2023}};
	\node [] (17) at (0.5, -2.25) {\scriptsize sjf-$\CQneg$ dichotomy \cite{reshefImpactNegationComplexity2020}};
	\node [rotate=90, align=center, anchor=center,text width = 4em] (34) at (-3.5, -2){\begin{spacing}{0.5}Shapley value\end{spacing}};
	\node [rotate=90, align=center, anchor=center,text width = 5em] (35) at (-3.5, 0) {\begin{spacing}{0.5}model counting\end{spacing}};
	\node [rotate=90, align=center, anchor=center,text width = 5em] (36) at (-3.5, 2.25) {\begin{spacing}{0.5}probabilistic evaluation\end{spacing}};
	\node [align=center, anchor=north,text width = 7em] (37) at (-2.2, 0.2) {\begin{spacing}{0.5}
			\scriptsize hardness for "unbounded" "hom-closed" "graph queries" \cite{amarilliUniformReliabilityUnbounded2023a}
		\end{spacing}};
	\node (37a)  at (37.north) [below=.03] {};

	\draw [color=puregreen, thick] (6) to (14);
	\draw [color=puregreen, thick] (14) to (7);
	\draw [color=puregreen, thick] (7) to (13);
	\draw [color=puregreen, thick] (13) to (1);
	\draw [color=puregreen, thick] (15) to (12);
	\draw [color=puregreen, thick] (16) to (12);
	\draw [color=puregreen, thick] (17) to (12);
	\draw [color=puregreen, thick] (37a.center) to (0);
	\draw [->, >= stealth, thick] (5) to (6);
	\draw [->, >= stealth, thick] (6) to (4);
	\draw [->, >= stealth, thick] (4) to (7);
	\draw [->, >= stealth, thick] (5) to (3);
	\draw [->, >= stealth, thick] (3) to (4);
	\draw [->, >= stealth, thick] (0) to (1);
	\draw [->, >= stealth, thick] (1) to (10);
	\draw [->, >= stealth, thick] (0) to (9);
	\draw [->, >= stealth, thick] (9) to (10);
	\draw [->, >= stealth, thick, color=red] (11) to (9);
	\draw [->, >= stealth, bend left=60, looseness=1.25, thick, color=red, dashed] (9) to (11);
	\draw [->, >= stealth, bend left=60, looseness=1.25, thick, color=red, dashed] (10) to (12);
	\draw [->, >= stealth, thick] (12) to (10);
	\draw [->, >= stealth, bend right=75, thick] (12) to%
	node[midway,below right] {\scriptsize\cite{deutchComputingShapleyValue2022a}}
	(7);
	\draw [->, >= stealth, thick] (11) to (12);
	\draw [<->, >= stealth, thick, color=red] (9) to (3);
	\draw [<->, >= stealth, thick, color=red] (10) to (4);
	\draw [<->, >= stealth, thick] (1) to (6);
	\draw [<->, >= stealth, thick] (0) to (5);
\end{tikzpicture}
 }
		\caption{}
		\label{fig:reductions}
	\end{subfigure}
	\begin{subfigure}{0.51\linewidth}
		\centering
		\definecolor{myblue}{HTML}{b3cde3}
\definecolor{mygreen}{HTML}{ccebc5}
\definecolor{myred}{HTML}{fbb4ae}

\resizebox{\linewidth}{!}{
\begin{tikzpicture}[xscale=.8,yscale=.8]
	\footnotesize
	\begin{pgfonlayer}{nodelayer}
		\node [draw, rectangle, rounded corners, align=center, fill=myblue] (0) at (-1.5, 2) {"dss" "hom-closed"\\on "graphs"};
		\node  at (0.east) [right=-.05] {\color{red}$\star$};
		\node [draw, rectangle, rounded corners, fill=myblue] (1) at (1.75, 3) {"dss" $C$-"hom-closed"};
		\node [draw, rectangle, rounded corners, fill=myblue] (2) at (.75, 4) {"RPQ"};
		\node  at (2.east) [right=-.05] {\color{green}$\star$\color{red}$\star$};
		\node [draw, rectangle, rounded corners] (3) at (1.75, 5) {$C$-"hom-closed"};
		\node [rectangle, rounded corners] (4) at (-2.3, 3) {\phantom{cc-disjoint-CRPQ}};
		\node [] (4a) at (4.south) [left=.4] {};
		\node [draw, rectangle, rounded corners,
		pattern={Lines[angle=45,distance=8pt,line width=4pt]},pattern color=myblue,
		preaction={fill=mygreen}
		] at (4) {\phantom{"cc-disjoint-CRPQ"}};
		\node [rectangle, rounded corners] (4) at (-2.3, 3) {"cc-disjoint-CRPQ"};
		\node  at (4.east) [right=-.05] {\color{red}$\star$};
		\node [draw, rectangle, rounded corners, fill=myblue] (5) at (-3.75, 4) {"conn." "hom-closed"};
		\node [draw, rectangle, rounded corners] (6) at (-2, 5) {"hom-closed"};
		\node [draw, rectangle, rounded corners,
		pattern={Lines[angle=45,distance=8pt,line width=4pt]},pattern color=myblue,
		preaction={fill=mygreen}] (7) at (-3, 1) {"sjf-CRPQ"};
		\node  at (7.east) [right=-.05] {\color{red}$\star$};
		\node [draw, rectangle, rounded corners, align=center, fill=myblue] (8) at (-4.75, 2) {"conn." "hom-closed"\\on "graphs"};
		\node  at (8.east) [right=-.05] {\color{red}$\star$};
		\node [draw, rectangle, rounded corners, fill=myblue] (9) at (-4.75, 0) {"conn." "UCRPQ"};
		\node  at (9.east) [right=-.05] {\color{red}$\star$};
		\node [draw, rectangle, rounded corners, fill=myblue] (10) at (-5.5, 3.0) {"conn." "UCQ"};
		\node  at (10.east) [right=-.05] {\color{red}$\star$};
		\node [rectangle, rounded corners,
		pattern={Lines[angle=45,distance=8pt,line width=4pt]},pattern color=myblue,
		preaction={fill=mygreen}] (11) at (-1.3, 0) {\phantom{"sjf-CQ"}};
		\begin{scope}
			\clip[rounded corners] (11.south west) rectangle (11.north east);
			\fill[color=myred] (11.south west) rectangle (11.north);
		\end{scope}
		\node [draw, rectangle, rounded corners] at (-1.3, 0) {"sjf-CQ"};
		\node  at (11.east) [right=-.05] {\color{green}$\star$\color{red}$\star$};
		\node [draw, rectangle, rounded corners, fill=myred] (12) at (-1.3, 1) {"CQ"};
		\node  at (12.east) [right=-.05] {$\dagger$};
		\node [draw, rectangle, rounded corners, fill=myred] (13) at (1.75, 1) {"sjf-CQ"};
		\node  at (13.east) [right=-.05] {{\color{green}$\star$\color{red}$\star$}$\dagger$};
		\node [] (14) at (0, -0.25) {};
		\node [] (15) at (0, 6) {};
		\node [] (16) at (-6.5, -0.5) {};
		\node [] (17) at (-6.5, 6) {};
		\node [] (18) at (3.5, -0.5) {};
		\node [] (19) at (3.5, 6) {};
		\node [] (20) at (0, 6) {};
		\node [below left, color=gray] (21) at (0, 6) {\textit{without "constants"}};
		\node [below right, color=gray] (21) at (0, 6) {\textit{with "constants"}};
	\end{pgfonlayer}
	\begin{pgfonlayer}{edgelayer}
		\draw [->, >=stealth,  thick] (11) to (12);
		\draw [->, >=stealth,  thick] (11) to (13);
		\draw [->, >=stealth,  thick] (1) to (3);
		\draw [->, >=stealth,  thick] (2) to (3);
		\draw [->, >=stealth,  bend right=75, looseness=.8, thick] (13) to (3);
		\draw [->, >=stealth,  thick] (7) to (4a.center);
		\draw [->, >=stealth,  thick] (9) to (8);
		\draw [->, >=stealth,  thick] (8) to (5);
		\draw [->, >=stealth,  thick] (10) to (5);
		\draw [->, >=stealth,  thick] (4) to (6);
		\draw [->, >=stealth,  thick] (5) to (6);
		\draw [->, >=stealth,  thick] (6) to (3);
		\draw [->, >=stealth,  thick] (0) to (1);
		\draw [color=lightgray, dashed] (18) to (19);
		\draw [color=lightgray, dashed] (19) to (17);
		\draw [color=lightgray, dashed] (17) to (16);
		\draw [color=lightgray, dashed] (16) to (18);
		\draw [color=lightgray, dashed] (14) to (15);
	\end{pgfonlayer}
\end{tikzpicture}
}
\resizebox{\linewidth}{!}{
\small
\begin{tabular}{lll}
	{\color{myblue}$\blacksquare$} \Cref{lem:pseudoconn} & \multicolumn{2}{l}{{\color{green}$\star$}/{\color{red}$\star$} = dichotomy exists (prior work/this paper)}\\
	{\color{myred}$\blacksquare$} \Cref{lem:leak} & \multicolumn{2}{l}{"dss" = "duplicable singleton support"}                                                                                                \\
	{\color{mygreen}$\blacksquare$} \Cref{lem:decomposable} & \multicolumn{2}{l}{$\dagger$= not preserving $\FGMC{}$-equiv.\ with the same query}
\end{tabular}
}
 		\caption{}
		\label{fig:query-lang}
	\end{subfigure}
	\hspace{-.1\linewidth}
	\vspace{-2ex}
	\caption{
		(a) A (clickable) summary of the reductions. An arrow from  $A$ to $B$ means a "polynomial-time Turing-reduction" from $A$ to $B$. Red arrows indicate our contributions.\\
		\phantom{}\hspace{9.25mm}(b) A (clickable) summary of the reductions from $\FGMC{}$ for the classes of queries captured by our results, and the "FP"/"shP"-hard  dichotomies that follow. The colors indicate the lemma(s) used for the proof: alternating diagonal stripes indicate that both lemmas are needed for the proof, while the vertical separation of "sjf-CQ" indicates two alternative ways to obtain the result.
	}
	\label{fig:figures}
%
\end{figure}
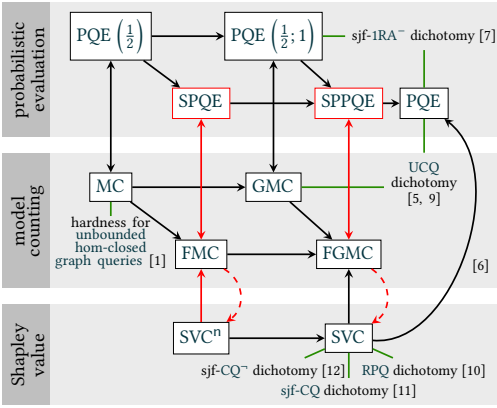
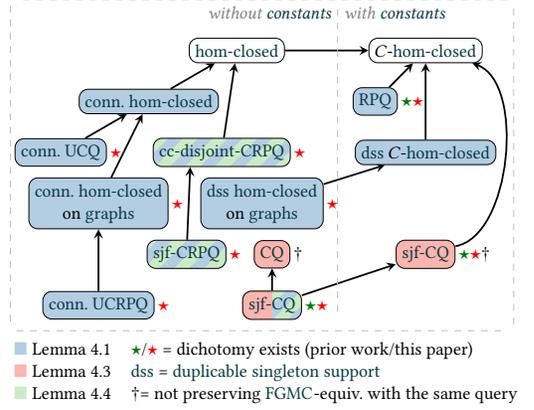


Can hardness results for $\Shapley{}$ similarly be explained via $\PQE{}$ or $\GMC{}$?
A recent result \cite{karaShapleyValueModel2023} further demonstrated the closeness 
of the settings by reproving the hardness of $\Shapley{}$ using 
the same model counting problem for Boolean functions 
as was originally used to show hardness of $\PQE{}$ for non-"hierarchical" "sjf-CQs" \cite{dalviEfficientQueryEvaluation2004}. 
However, a reduction from $\PQE{} $ or $\GMC{}$ to $\Shapley{}$ 
has yet to be exhibited, no doubt due to the fact that 
$\Shapley{}$ outputs a complex weighted sum which cannot be obviously employed for probability or model counting computations. 
Our paper investigates this missing link and shows that the matching dichotomy conditions are no 
coincidence, as it is possible to exhibit reductions from $\PQE{} $ and $\GMC{}$ to $\Shapley{}$
for broad classes of queries, 
which 
not only explain existing hardness results for $\Shapley{}$, but also yield some new complexity results.

%
%
%

\subsection*{Contributions}
Our first conceptual contribution (Section \ref{sec:problems}) is to introduce suitable specializations of the $\PQE{}$ problem, namely $\SPQE{}$ and $\SPPQE{}$, obtained by restricting the probability values of facts,
and showing that they are polynomial-time equivalent to variants $\FMC{}$ and $\FGMC{}$  of the model counting $\MC{}$ and generalized model counting $\GMC{}$ problems ("cf" ``model counting'' and ``probabilistic evaluation'' boxes of \Cref{fig:reductions}). 
While some of these problems had been implicitly used in proofs, or defined for related settings, their formal introduction 
in the database setting serves to bring some clarity and order to the literature 
and provides us with suitable analogues of $\Shapley{}$ in the probabilistic and counting domains. 
Indeed, the existing reduction of $\Shapley{}$ to $\PQE{}$ 
can more precisely be expressed as a reduction to the more restricted $\FGMC{}$ and $\SPPQE{}$ problems.

Our main technical contributions (presented in Section \ref{sec:mainresults}
and proven in Section \ref{sec:mainproofs})
are the first known reductions in the opposite direction:
from $\FGMC{}$  (equivalently, $\SPPQE{}$) to $\Shapley{}$. 
We in fact provide three such reductions
(\Cref{lem:pseudoconn,lem:leak,lem:decomposable}),
each employing somewhat different structural conditions and 
covering different classes of ($C$-)homomorphism-closed queries. 
\changed{Importantly\siderev{Added comment to highlight query-preserving reductions and obtained PTIME equivalences.}, in many cases our reductions preserve the query, i.e.\  reduce $\FGMC{\changed{q}}$ to $\Shapley{\changed{q}}$\footnote{\changed{As we consider data complexity, the query is treated as fixed, 
so each query $q$ will give rise to a separate SVC / (F)(G)MC problem (see Section \ref{sec:problems} for formal definitions). }}.
We thereby establish the polynomial-time equivalence of $\FGMC{\changed{q}}$ and $\Shapley{\changed{q}}$ for a range of query classes,
providing strong new evidence that SVC is nothing more than a counting problem. }



As a consequence of our reductions, it is now possible to obtain the 
"FP"/"shP"-hard dichotomies for $\Shapley{}$ for "sjf-CQ"s and "regular path queries"
as corollaries of existing dichotomies 
for $\PQE{}$ and $\GMC{}$ (intuitively, by following the arrows in Figure \ref{fig:reductions} from and to
known results). 
Moreover, due to the general nature of our lemma hypotheses, additional new 
results can be obtained. In particular, we establish "FP"/"shP" dichotomies
for constant-free unions of connected CQs and connected homomorphism-closed graph queries. 
These results are presented in Section  \ref{sec:mainresults} and summarized in Figure \ref{fig:query-lang}.

In Section \ref{sec:extensions}, we further demonstrate the utility of our proof techniques by showing how they can be applied to other scenarios. 
  %
  First, we show that essentially the same reductions, modulo an additional hypothesis, can be 
  used to show hardness of Shapley value computation over databases without assumed (``"exogenous"'') facts
  (existing reductions make the crucial use of the presence of ``exogenous'' facts).
 Second,  a simple adaptation of our techniques allows us to 
 partially recapture existing results for 
"CQs with safe negation@self-join-free conjunctive queries with safe negations" \cite{reshefImpactNegationComplexity2020}. 
Thirdly, we show that computing the maximum Shapley value is just as hard as $\Shapley{}$ as  far as our reductions go.
Finally, we initiate the study of Shapley value of database constants (rather than facts) and employ our methods to show how this problem can be formally related to the analogous model counting problems. 

\section{Preliminaries}
\label{sec:prelim}


\AP
We fix disjoint infinite sets $\intro*\Const$, $\intro*\Var$ of ""constants"" and ""variables"", respectively. 
\AP
\changed{We write $A \intro*\dcup B$ to denote the union $A \cup B$ of two disjoint sets $A,B$.}
For any syntactic object $O$ ("eg" database, query), we will use $\intro*\vars(O)$ and $\intro*\const(O)$ 
to denote the sets of "variables" and "constants" contained in $O$, and let $\intro*\mterms(O)\defeq \vars(O) \cup \const(O)$ denote its set of ""terms"". 
\AP A ""(relational) schema"" 
is a finite set of relation symbols, each associated with a (positive) arity. 
\AP
A ""(relational) atom"" over a "schema" $\Sigma$ takes the form $R(\bar t)$ where $R$ is a ""relation name"" from $\Sigma$ of some arity $k$, and $\bar t \in (\Const \cup \Var)^k$.\
\AP
A ""fact"" is an "atom" which contains only "constants".
\AP
A ""database"" $\D$ over a "schema" $\Sigma$ is a finite set of "facts" over $\Sigma$, and we call it a 
""graph database"" if $\Sigma$ is a binary schema ("ie",\ 
consisting 
of binary relations only). 


\AP
For sets of "atoms" $S_1,S_2$, a ""homomorphism"" from $S_1$ to $S_2$ is a function 
$h: \mterms(S_1) \rightarrow  \mterms(S_2)$ 
such that $R(h(t_1), \dotsc, h(t_k)) \in S_2$ for every 
$R(t_1, \dotsc, t_k) \in S_1$. We write $S_1 \intro*\homto S_2$ to indicate the existence of such $h$. 
\AP
If further we have that $h(c) = c$ for every $c \in C \subseteq \Const$, 
we call $h$ a ""$C$-homomorphism"" and write $S_1 \intro*\Chomto S_2$.

\AP
We say that a set of  "atoms" $S$ is ""connected""  if so is its associated undirected 
""incidence
graph"" 
$\intro*\incgraph S$, whose nodes are $S \cup \mterms(S)$ and whose edges are $\{\, \{t,\alpha\} \mid t \in \mterms(\alpha), \alpha \in S\}$.

\paragraph{Query languages}\AP 
To simplify the presentation, we shall only consider Boolean queries, which we will often refer to simply as ``queries''. 
\AP
A ""Boolean query"" $q$ specifies a true-or-false property of databases. 
\AP
We write $\D \models q$ to denote that database $\D$ satisfies query~$q$, in which case we may also call $\D$ a ""support"" for $q$. 

\AP
A query $q$ is ""closed under homomorphisms"", or \reintro{hom-closed}, if for every pair $\D,\D'$ of databases 
such that $\D \homto \D'$ and $\D \models q$, we have $\D' \models q$ as well. We define similarly ""closure under $C$-homomorphisms"", or \reintro{$C$-hom-closed}, to capture queries with constants. 
For ($C$-)hom-closed queries, it is natural to consider the smallest databases that satisfy the query. 
\AP Formally, we call $\D$ a ""minimal support"" for $q$
if $\D \models q$ and $\D' \not \models q$ for every $\D' \subsetneq \D$. 
\AP
We say a "fact" $\alpha$ is ""relevant to"" a query $q$ if it appears in some "minimal support" of $q$.
\AP A query $q$ is \reintro{connected} if every "minimal support" for $q$ is "connected".

We will consider queries formulated in (fragments of) standard query languages. 
\AP
A (Boolean) ""conjunctive query"" (henceforth just \reintro{CQ}) over $\Sigma$
is a conjunction of "atoms" over $\Sigma$, all of whose variables are existentially quantified. 
We use $\intro*\atoms(q)$ to denote the set of "atoms" of a "CQ"~$q$. 
\AP
A database $\D$ ""satisfies"" a "CQ" $q$ if  $\atoms(q) \Chomto \D$ where $C=\const(q)$.
\AP We call a CQ $q$ ""self-join-free"" if no two "atoms" of $q$ have the same "relation name", and we use ""sjf-CQ"" for the class of all "self-join-free" "CQs".
\AP A ""union of conjunctive queries"" (\reintro{UCQ}) is a finite disjunction of CQs, with  $\D \models q_1 \vee \dotsb \vee q_\ell$
 iff 
$\D \models q_i$ for some $1 \leq i \leq \ell$. 
\AP ""Infinitary unions of CQs"" (\reintro{UCQ$^\infty$}) are defined similarly but the disjunction may be over a countably infinite set of CQs.

\AP
A ""graph query"" is any query over a binary "schema".
\AP 
A ""path atom"" over a binary "schema"~$\Sigma$ takes the form $L(t,t')$, with $t,t' \in \Const \cup \Var$
and $L$ a regular expression over the alphabet~$\Sigma$.
A (Boolean) ""regular path query"", or \reintro{RPQ}, over $\Sigma$ is  
a "path atom" 
$L(a,b)$ over $\Sigma$, with $a,b \in \Const$. Satisfaction of RPQs is 
as follows:
$\D \models L(a,b)$
 iff there exists a word $w = R_1 \ldots R_\ell$ in the language of $L$
 and a sequence $c_0, c_1, \ldots, c_{\ell - 1}, c_\ell$ of constants 
 such that $c_0=a$, $c_\ell=b$, and $R_i(c_{i-1}, c_i) \in \D$ for every $1 \leq i \leq \ell$. 
\AP
A (Boolean) ""conjunctive regular path query"" (\reintro{CRPQ}) over 
$\Sigma$
is an existentially quantified conjunction of "path atoms" over $\Sigma$. 
A "graph database" $\D$ satisfies a CRPQ $q$ if there exists a mapping $h:\mterms(q) \rightarrow \const(\D)$
such that $h(c)=c$ for every $c \in \const(q)$ and $\D \models L(h(t),h(t'))$  for every "path atom" $L(t,t')$ of $q$. 
We shall also consider 
""unions of conjunctive regular path queries"" (\reintro{UCRPQs}), defined as expected. 

\paragraph{Reductions}
\AP
We will use the standard notion of polynomial-time Turing-reductions between numeric problems $\Psi_1,\Psi_2$. Concretely, we write $\Psi_1 \intro*\polyrx \Psi_2$, and say that there is a ""polynomial-time reduction"" from $\Psi_1$ to $\Psi_2$, if there exists a polynomial-time algorithm for computing $\Psi_1$ using unit-cost calls to $\Psi_2$.
If $\Psi_1 \polyrx \Psi_2$ and $\Psi_2 \polyrx \Psi_1$, we write $\Psi_1 \intro*\polyeq \Psi_2$.

 \section{Studied Problems}\label{sec:problems}

In this section, we formally introduce Shapley value computation, which is the focus of the paper, as well as related model counting and probabilistic querying problems. 

\AP
For the definition of Shapley values, as well as the task of generalized model counting, 
we shall require that the input database $\D$ has been partitioned into two sets 
$\Dn$ and $\Dx$ of ""endogenous"" and ""exogenous"" "facts" respectively. 
\AP
The notation $\D = (\intro*\Dn, \intro*\Dx)$ 
will be used to 
refer to such ""partitioned databases"". All databases will henceforth be assumed to 
be partitioned unless otherwise noted. 

\subsection{Shapley Value Computation}
The "Shapley value" has been introduced by Lloyd Shapley in 1952  \cite{shapley:book1952} to distribute the wealth of a "cooperative game", and it follows three axioms of good behaviour. 
Intuitively, these axioms state that the value of two isomorphic "games" is the same, that the sum of all players values equals the total wealth of the "game" and that the value of a sum of two "games" equals the sum of the individual games' values. 
Remarkably, there is only one function satisfying these axioms: the Shapley value.

\AP
A ""cooperative game"" is given by a finite set of players $P$ and a "wealth function" $\scorefun[] : \partsof P \to \IQ$ that represents the wealth generated by every possible coalition of players, such that $\scorefun[](\emptyset)=0$.
\AP
To define the "Shapley value", assume 
players arrive one by one in a random order, and each earns the difference between the coalition's wealth before she arrives and after. 
The ""Shapley value"" of a player $p\in P$ is her expected earnings in this scenario, which can be expressed as: 
\AP
\begin{equation}\label{def_sh}
	\intro*\Sh(P,\scorefun[],p) \defeq \frac{1}{|P|!}\sum_{\sigma\in \Sym(P)} \left(\scorefun[](\sigma_{<p} \cup \{p\}) - \scorefun[](\sigma_{<p})\right)
\end{equation}
\AP
where $\intro*\Sym(P)$ denotes the set of permutations of $P$ and $\sigma_{<p}$ the set of players that appear strictly before $p$ in the permutation $\sigma$. The following equivalent formula -- 
more convenient for our proofs -- can be obtained by grouping together 
 the $\sigma$ having the same $\sigma_{<p}$.

\begin{equation}\label{formul_sh}
	\reintro*\Sh(P, \scorefun[], p) = \sum_{B\subseteq P \setminus \set p } \!\!\!  \frac{|B|!(|P| - |B| -1)!}{|P|!}\left(\scorefun[](B \cup \{p\}) - \scorefun[](B)\right)
\end{equation}

\AP
The "cooperative games" that will interest us 
are those associated with a "Boolean query" $q$ and a "partitioned database" $\D = (\Dn, \Dx)$, where the set of players 
is the set $\Dn$ of "endogenous" "facts", and the ""wealth function"" 
$\intro*\scorefun$ assigns $v_S-v_{\subexo}$ to each subset $S \subseteq \Dn$, where $v_S=1$ (resp.\ $v_{\subexo}=1$) if $S \cup \Dx \models q$ (resp.\ if $\Dx \models q$), and $0$ otherwise. 

The main focus of this paper is to understand the complexity of computing the "Shapley value" of facts in such "games". 
For a fixed "Boolean query" $q$, this is the task of computing, for a given input "database" $\D$ and "fact" $\alpha \in \D_n$, the value  
$\Sh(\Dn, \scorefun, \alpha)$. 
\AP
We will use $\intro*\Shapley q$ to refer to this computational task, and $\intro*\Shapleyn q$ to the task when restricted to "partitioned databases" with only "endogenous" "facts" ("ie", of the form $\D = (\Dn, \Dx)$ with $\Dx = \emptyset$).

\begin{remark}[Non-Boolean case] Shapley value computation could be defined for non-Boolean queries by asking 
for the contribution of a fact to a given tuple being a query answer. 
This variant straightforwardly reduces to our Boolean version, by substituting the query's free variables with constants
from the answer tuple. Observe however that this reduction yields Boolean queries \emph{with constants}, thus motivating
the interest of obtaining results for queries with constants. \end{remark}

\subsection{Model Counting}
We now introduce several versions of the model counting problem.
\AP
For any "Boolean query" $q$, the ""generalized model counting problem"" on $q$, or $\intro*\GMC q$ problem, is the task of computing, for a given input "database" $\D = (\Dn, \Dx)$,
the number of subsets $S \subseteq \Dn$ such that $S \dcup \Dx \models q$ (such subsets will be called ""generalized supports for $q$ in  $\D$""). 
\AP
On the other hand, the ""fixed-size GMC problem"", or $\intro*\FGMC q$, is the task of computing, for an \siderev{Improved definition of $\FGMC{}$}\changed{input number} $n \in \Nat$ and "database" $\D = (\Dn, \Dx)$, the number of subsets $S \subseteq \Dn$ of size exactly $n$ such that $S \dcup \Dx \models q$.\footnote{\changed{We emphasize that in $\FGMC q$ the size $n$ is part of the input (the name of ``fixed-size'', coined in \cite{karaShapleyValueModel2023} in a slightly different context, may lead to confusion).}}
\AP
The ""model counting problem"" $\intro*\MC q$ and ""fixed-size model counting problem"" $\intro*\FMC q$ correspond to the previous problems when restricted so that the input "database" contains no "exogenous" "facts" ("ie", $\Dx =\emptyset$). 
%
%

\subsection{Probabilistic Query Evaluation}
\AP
A ""tuple-independent probabilistic database"" is a pair $\D = (S,\pi)$ where $S$ is a finite set of "facts" and $\pi: S \to (0,1]$ is a probability assignment.
\AP
We ""associate@@pdb"" to every such database a "partitioned database" with $\Dx{}\defeq \{\alpha\in S \mid \pi(\alpha)=1\}$. 
For a "Boolean query"~$q$, $\intro*\Proba(\D \models q)$ is the probability of $q$ being true, where each "fact" $\alpha$ has independent probability $\pi(\alpha)$ of being in the "database". 
\AP
For a fixed "Boolean query"~$q$, the problem of computing, for a given a "tuple-independent probabilistic database" $\D$,  the probability $\Proba(\D \models q)$ is known as the ""probabilistic query evaluation"" problem, or $\intro*\PQE q$.
\AP
We consider also restrictions of $\PQE q$ where the input "probabilistic database" $(S,\pi)$ is such that the image $Im(\pi)$ has certain characteristics.
\begin{enumerate}
	\item If $Im(\pi) = \set{\frac 1 2}$ we obtain $\intro*\PQEPhalf q$.
	\item If $Im(\pi) \subseteq \set{\frac 1 2, 1}$ we obtain $\intro*\PQEPhalfOne q$.
	\item If $Im(\pi) = \set p$ for some $p \in (0,1]$, we obtain $\intro*\SPQE q$: the ""single probability query evaluation"" problem.
	\item If $Im(\pi)=\set{p,1}$ for some $p \in (0,1]$, we obtain $\intro*\SPPQE q$: the ""single proper probability query evaluation"" problem.
\end{enumerate}

\subsection{Prior Work and Relations Between Problems}
\AP
$\PQE{}$ was first introduced in \cite{dalviEfficientQueryEvaluation2004}, and the most important result about it is the "FP"/"shP"-hard dichotomy for "UCQs", established in \cite{dalviDichotomyProbabilisticInference2012}. The "UCQs" for which $\PQE{}$ is tractable are known as ``""safe""''.
$\MC{}$, $\GMC{}$ and their probabilistic counterparts have been considered in \cite{kenigDichotomyGeneralizedModel2021} where the same dichotomy has been extended to the latter:
\begin{proposition}[{\cite[Theorem~4.21]{dalviDichotomyProbabilisticInference2012} and \cite[Theorem~2.2]{kenigDichotomyGeneralizedModel2021}}] Let $q$ be a "Boolean" "UCQ". If $q$ is "safe", then both $\PQE q$ and $\GMC q$ are in "FP", otherwise both are "shP"-hard.
\end{proposition}

\AP
Another notable result is the hardness of $\MC{}$ for "hom-closed" "graph queries" which are ""unbounded"", 
that is, queries which are not equivalent to a "UCQ":
\begin{proposition}[{\cite[Theorem~1.3]{amarilliUniformReliabilityUnbounded2023a}}]
	 For every "unbounded" "hom-closed" "graph query" $q$, $\MC q$ is "shP"-hard.
\end{proposition}
This result immediately yields a "FP"/"shP" dichotomy for $\GMC{}$ over "hom-closed" "graph queries".
To the best of our knowledge, $\FGMC{}$, $\FMC{}$, $\SPPQE{}$, and $\SPQE{}$ have not been explicitly introduced before in the context of databases. However, a notion of fixed-size model counting over Boolean functions has been used in \cite{karaShapleyValueModel2023},
 which is analogous to $\FGMC{}$.\footnote{The reason why it is analogous to $\FGMC{}$ rather than to $\FMC{}$ is because the "exogenous" "facts" get abstracted out in the transformation between databases and Boolean functions.}

\AP
$\Shapley{}$ was first introduced in \cite{livshitsShapleyValueTuples2021}, 
where a "FP"/"shP"-hard dichotomy was established for "sjf-CQs".
 Interestingly, the latter result identifies the same tractable queries as for $\PQE{}$, namely, the "safe" "sjf-CQs", 
 which admit a simple syntactic characterization as the \reintro{hierarchical}\footnote{We recall that a "CQ" $q$ is \emph{not} ""hierarchical"" if{f} there exist $\alpha_1, \alpha_2, \alpha_3 \in \atoms(q)$ such that 
$(\vars(\alpha_1)\cap \vars(\alpha_2))\not \subseteq \vars(\alpha_3)$ and $(\vars(\alpha_3)\cap \vars(\alpha_2))\not \subseteq\vars(\alpha_1)$.} "sjf-CQs". 
Recent work has clarified the relation between the two dichotomies
by reducing $\Shapley{}$ to $\PQE{}$ \cite{deutchComputingShapleyValue2022a}
and reproving the hardness of $\Shapley{}$ \cite{karaShapleyValueModel2023} 
by reduction from the same model counting problem for Boolean functions  
that had been used to show hardness of $\PQE{}$ for non-"hierarchical" "sjf-CQs" \cite{dalviEfficientQueryEvaluation2004}. 

The problems presented in this section are closely related. \Cref{fig:reductions} summarizes the reductions between these problems. 


\begin{propositionrep}\label{prop:reductions-problems}
	For any "Boolean query" $q$, all the "polynomial-time reductions" denoted by the solid arrows of \Cref{fig:reductions} hold. In particular:
	\begin{enumerate}
		\item\label{cas1.1} $\FGMC q \polyeq \SPPQE q$;
		\item\label{cas1.2} $\FMC q \polyeq \SPQE q$;
		\item\label{cas1.3} $\Shapley q \polyrx \FGMC q$.
	\end{enumerate}
	Moreover, the reductions for (\ref{cas1.1}) and (\ref{cas1.2}) only use the oracle on inputs with the same "associated partitioned database".
\end{propositionrep}
\begin{proofsketch}
	The fact that $\Shapley q \polyrx \PQE q$ is known from \cite[Proposition~3.1]{deutchComputingShapleyValue2022a}, and the proofs of \cref{cas1.1,cas1.2,cas1.3} are adaptations of this result
	(details in the appendix). 
	The proof of $\Shapleyn{q} \polyrx \FMC q$ is relegated to \Cref{sec:endogenous}.
The remaining reductions are trivial.
\end{proofsketch}
\begin{appendixproof}
	Since the other reductions were justified in the main body, we shall only prove here the three remaining reductions stated above, by adapting the proof of \cite[Proposition~3.1]{deutchComputingShapleyValue2022a}. For clarity we will split them into several lemmas.

\begin{claim}
	For every Boolean query $q$, $\Shapley q \polyrx \FGMC q$.
\end{claim}

\begin{proof}
	Recall the formula for $\Sh$:
	
	\[ 
	\Sh(\Dn ,v,\mu) = \sum_{B\subseteq \Dn \setminus\{\mu\}} \frac{|B|!(|\Dn |-|B|-1)!}{|\Dn |!} (\scorefun(B\cup\{\mu\}) - \scorefun(B)) 
	\]
\noindent	Now denote by $n$ the size of $|\Dn|$ and by $\mathrm{FGMC}_j(q)(\Dn,\Dx)$ the answer of $\FGMC q$ for the "partitioned database" $(\Dn,\Dx)$. Also abbreviate $j!(|\Dn |-j-1)!/|\Dn |!$ as $C_j$. Grouping the $B$s by their sizes, the formula becomes as follows (in the first $\mathrm{FGMC}$ term $\mu$ becomes "exogenous" and in the second it is removed).
	
	\[\Sh(\Dn ,v,\mu) = \sum_{j=0}^{n} C_j \left[\mathrm{FGMC}_j(q)(\Dn \setminus\{\mu\},\Dx \cup\{\mu\})
	-\mathrm{FGMC}_j(q)(\Dn \setminus\{\mu\},\Dx )\right]\]
	
%
\noindent	Assuming an oracle to $\FGMC q$, one can thus obtain $\Sh(\Dn ,v,\mu)$ by straightforward arithmetic computations.
\end{proof}

\begin{claim}\label{lem:sppqe}
	For every Boolean query $q$, $\FGMC q \polyeq \SPPQE q$. Moreover, the reductions preserve the underlying "partitioned database".
\end{claim}

\begin{proof}
	We keep the same notations as in the previous proof.\\
	
\noindent	$\SPPQE q \polyrx \FGMC q$:
	Given a probabilistic database $\D$ where all facts have probabilities in $\{p;1\}$, denote by $(\Dn,\Dx)$ its underlying "partitioned database". Recall that $\Dx\defeq\{f\in\D|\pi(f)=1\}$ and $\Dn\defeq\D\setminus\Dn$. By denoting $z\defeq \frac{p}{1-p}\in\IQ_+$, the proper probability can be written $\frac{z}{1+z}$. Then the following relation can be obtained without much difficulty:
	
	\[
	(1+z)^{n}\Pr(\D\models q) = \sum_{j=0}^{n} z^j \mathrm{FGMC}_j(q)(\Dn ,\Dx )
	\]
	This means that $\Pr(\D\models q)$ can be computed in polynomial time with access to an oracle to $\FGMC q$, and only making oracle calls over the same underlying "partitioned database".\\
	
\noindent	$\FGMC q \polyrx \SPPQE q$:
	We wish to compute $\mathrm{FGMC}_j(q)(\Dn ,\Dx )$ for some input database $\D=\Dn \dcup \Dx$. Denote $n\defeq|\Dn |$ and consider for all $z\in\IQ_+$ the probabilistic database $\D_z$ whose underlying database is $(\Dn,\Dx)$ and where every fact in $\Dn$ has probability $\frac{z}{1+z}$. Then the following relation can be obtained as before:
	
	\[
	(1+z)^{n}\Pr(\D_z\models q) = \sum_{j=0}^{n} z^j \mathrm{FGMC}_j(q)(\Dn ,\Dx )
	\]
	
	From there one can obtain a system of linear equations by using the $\SPPQE q$ oracle on $\D_z$ for $n+1$ distinct values of $z$. Since the underlying matrix is a Vandermonde with distinct coefficients, it is invertible hence the system can be solved to obtain in polynomial time the different $\mathrm{FGMC}_j(q)(\Dn ,\Dx )$.
\end{proof}

\begin{claim}
	For every Boolean query $q$, $\FMC q \polyeq \SPQE q$. Moreover, the reductions preserve the underlying "partitioned database".
\end{claim}

\begin{proof}
	It is a direct consequence of \Cref{lem:sppqe}: since the underlying "partitioned database" is preserved, the absence of "exogenous" "facts" in the counting instance will be equivalent to the absence of deterministic tuples in the probabilistic ones. 
\end{proof}

\phantom\qedhere \end{appendixproof}
\section{Main Results}\label{ssec:pcnx}
\label{sec:mainresults}

Our main results are reductions from $\FGMC{q}$ to $\Shapley{q}$ for large classes of "$C$-hom-closed" queries, which shall enable us to obtain known and new dichotomy results for $\Shapley{}$ for several classes of queries. At a very high-level we show that:
\begin{enumerate}[(i)]
    \item $\FGMC{q} \polyrx  \Shapley{q}$ when $q$ is (almost) connected (\Cref{lem:pseudoconn}),
    \item $\FGMC{q} \polyrx  \Shapley{q \land q'}$ if $q$ is connected and the queries $q$ and $q'$ do not have an undesirable interaction (\Cref{lem:leak}),
    \item $\FGMC{q} \polyrx  \Shapley{q}$ whenever $q$ can be decomposed into two queries $q = q_1 \land q_2$ with no undesirable interaction (\Cref{lem:decomposable}).
\end{enumerate}
However, the concrete notions of what we mean by ``connected'' and ``undesirable interaction'' are somewhat technical and differ between the items.
The rationale for such definitions is to make the reduction statements as general as possible, which will become evident when we present the proofs in \Cref{sec:mainproofs}. The involved hypotheses are the price to pay for the ability to derive complexity results for several query languages at once.

In the remainder of the section, we make precise these hypotheses, provide formal statements of our key results, and present corollaries for concrete query classes (summarized in \Cref{fig:query-lang}). 

\subsection{Almost Connected Queries}
\AP
Consider a finite set 
$\intro*\aC \subseteq \Const$ 
and a "$\aC$-hom-closed" query $q$.  
\AP
An ""island support"" for $q$ is a support $S$ for $q$ such that for every set of facts $S'$ with $\const(S) \cap \const(S') \subseteq \aC$,
every "minimal support" for $q$ in $S \cup S'$ is contained either in $S$ or in $S'$. 
Intuitively, the "facts" of $S$ cannot participate in "minimal supports" outside $S$, except perhaps for those over $\aC$.
\AP
We say that $q$ is ""pseudo-connected"" if it has a "minimal support" $S$ that is an "island support" and such that $\const(S) \not\subseteq \aC$.


\begin{lemma}\label{lem:pseudoconn}
    For every "pseudo-connected" "$\aC$-hom-closed" query $q$, we have $\FGMC q \polyrx \Shapley q$.
\end{lemma}

The most notable class of "pseudo-connected" queries is the one of "connected" "hom-closed" queries 
. Observe that these are queries which are equivalent to infinitary unions of "connected" constant-free "CQs".

\begin{lemma}\label{lem:connected-is-pseudoconnected}
    Every "connected" "hom-closed" query is "pseudo-connected".
\end{lemma}
\begin{proof}
    Let $q$ be a "connected" "hom-closed" query; as far as the definition of "pseudo-connectedness" is concerned, this means  $\aC = \emptyset$. Take a "minimal support" $S$ of $q$. Let $S'$ be a set of facts with $\const(S)\cap\const(S')=\emptyset$, and $S^*$ a "minimal support" of $q$ in $S\cup S'$. It must be "connected" because $q$ is, and since there is no fact in $S\cap S'$ that contains both a constant in $\const(S)$ and one in $\const(S')$ (plus every relation has a positive arity), this implies it must be fully contained in either $S$ or $S'$. This makes $S$ is an "island support" for $q$ and $q$ "pseudo-connected".
\end{proof}

\begin{corollary}\label{cor:connected-homclosed}
    For every "connected" "hom-closed" query $q$, we have $\FGMC q \polyeq \Shapley q$.
\end{corollary}

As a consequence of \Cref{cor:connected-homclosed}, we obtain three "FP"/"shP"\siderev{Corrected from -complete}\changed{-hard} dichotomies (also depicted in \Cref{fig:query-lang}) by leveraging the known "UCQ"  dichotomy for $\GMC{}$ and $\PQE{}$ \cite{dalviDichotomyProbabilisticInference2012,kenigDichotomyGeneralizedModel2021}, as well as the hardness for $\MC{}$ over "hom-closed" graph queries \cite{amarilliUniformReliabilityUnbounded2023a}:

\begin{corollary}\label{cor:conn-hom-closed}\hfill
    \begin{enumerate}
        \item For any "connected" "hom-closed" "UCQ" $q$,  $\Shapley q$ is either in "FP" or "shP"-hard. 
        \item For any "connected" "hom-closed" "graph query" $q$,  $\Shapley q$ is either in "FP" or "shP"-hard. 
        ---in particular, this holds for any "connected" "UCRPQ" without "constants".
    \end{enumerate}
\end{corollary}
\begin{proof}
    $(1)$ Let $q$ be a "connected" "hom-closed" "UCQ". By \Cref{cor:connected-homclosed}, $\FGMC q \polyeq \Shapley q$. Now there can be two cases:
    \begin{itemize}
    	\item $q$ is safe, in which case $\FGMC q$ is tractable because $\PQE{q}$ is \cite{dalviDichotomyProbabilisticInference2012};
    	\item $q$ is unsafe, in which case $\FGMC q$ is intractable because $\GMC{q}$ is \cite{kenigDichotomyGeneralizedModel2021}.
    \end{itemize}

	\smallskip

	$(2)$ Let $q$ be a "connected" "hom-closed" "graph query" $q$. Again, by \Cref{cor:connected-homclosed}, $\FGMC q \polyeq \Shapley q$. Now there can be three cases:
	\begin{itemize}
		\item $q$ can be written as a safe "UCQ", in which case $\FGMC q$ is tractable because $\PQE{q}$ is \cite{dalviDichotomyProbabilisticInference2012};
		\item $q$ can be written as an unsafe "UCQ", in which case $\FGMC q$ is intractable because $\GMC{q}$ is \cite{kenigDichotomyGeneralizedModel2021};
		\item $q$ is unbounded, in which case $\FGMC q$ is intractable because $\MC{q}$ is \cite{amarilliUniformReliabilityUnbounded2023a}. \qedhere
	\end{itemize}
\end{proof}

Another class of "pseudo-connected" queries is the class of "RPQs" ("cf"~\Cref{lem:RPQ-is-pseudoconnected}). We can then obtain the dichotomy for "RPQ"s established in \cite{khalilComplexityShapleyValue2023} now as corollary of \Cref{lem:pseudoconn}.
\begin{toappendix}
    \begin{lemma}\label{lem:RPQ-is-pseudoconnected}
        Every (instantiated) "RPQ" 
        whose language contains a word of length at least 2 is "pseudo-connected".
    \end{lemma}
    \begin{proof}
        Let $q$ be an "RPQ". It can be written as a "path atom" $L(a,b)$ with $L$ a regular language and $a,b$ two constants that may be equal. $q$ is "$\aC$-hom-closed" for $\aC\defeq\{a,b\}$. Consider a word $w = R_1 \dots R_l$ of length at least 2, and a simple path $P\defeq a \xrightarrow{R_1} a_1 \dots a_{l-1} \xrightarrow{R_l} b$ such that $\forall i\in [l-1]. a_i \notin \aC$. $P$ is a "minimal support" of $q$ by definition, and $\const(P)\not\subseteq \aC$. We now show that it is an "island support".
        
        Take a set of facts $S'$ with $\const(P)\cap\const(S')\subseteq\aC$ and a minimal support $B\subseteq P\cup S'$ of $q$. It must be a path from $a$ to $b$, and any path that connects $a$ and $b$ via some edge in $P$ must contain $P$ as a whole, since all edges of $P$ contain some $a_i$ and no edge of $S'$ does. Therefore by minimality $B\subseteq P$ or $B\subseteq S'$.
    \end{proof}
\end{toappendix}

\begin{corollary}[Originally established in~\cite{khalilComplexityShapleyValue2023}]
	Let $q$ be an "RPQ" of language $L$. If $L$ contains a word of length at least 3, then $\Shapley{q}$ is "shP"-hard, otherwise it is in "FP".
\end{corollary}
\begin{proof}
	If $L$ contains no word of length at least 2, then $q$ can be written as a disjunction of atoms with no variable, which is trivially tractable for both $\FGMC{}$ and $\Shapley{}$ (rather than the whole input database, it suffices to consider the constant-size subset of facts that could appear in a minimal support). Otherwise, $q$ is "pseudo-connected" by \Cref{lem:RPQ-is-pseudoconnected}, hence \Cref{lem:pseudoconn} states that $\FGMC q \polyeq \Shapley q$.
	
	It now suffices to prove that the desired dichotomy holds for $\FGMC{}$. Tractability is easy: since all tractable queries in the dichotomy are bounded, the dichotomy for $\PQE{}$ \cite{dalviDichotomyProbabilisticInference2012} applies (see \cite{khalilComplexityShapleyValue2023} for a detailed explanation of why these queries are "safe"). As for the hardness, we turn to the proof of \cite[Theorem 4.1]{khalilComplexityShapleyValue2023} where they showed the hardness of $\Shapley{q}$ from the previously known hardness of $\Shapley{q_{RST}}$, with $q_{RST}\defeq R(x)\land S(x,y)\land T(y)$. As it turns out, their constructions work equally well to show that $\GMC{q_{RST}} \polyrx \GMC{q}$, and in turn the hardness of $\GMC{q_{RST}}$ is well known \cite{kenigDichotomyGeneralizedModel2021}.
\end{proof}

\AP The "pseudo-connected" class contains also queries that are \emph{not connected} in any real sense, for instance any $C$-"hom-closed" query with a ``""duplicable singleton support""'', that is, a "support" of size 1 that contains a "constant" outside of $C$. The name comes from the fact that we can make "homomorphic" copies by renaming a "constant" outside of $C$. The most basic example of such a query would be $A(x)\lor q$ for any "$C$-hom-closed" $q$. More interesting examples would include a "CRPQ" containing just one "path atom" $\exists x ~ L(a,x)$, where $a$ is a constant and $L$ any language containing a word of length 1, such as $L=A^*B$. Other examples may include
disconnected atoms, such as
$\exists x,y,u,z\,\, R(a,x,y)\land R(u,b,z)$ (with $a,b$ "constants").

\begin{corollary}\label{cor:singleton-sup}
	For every "$\aC$-hom-closed" query $q$ that admits a "duplicable singleton support", it holds that $\FGMC q \polyeq \Shapley q$.
\end{corollary}
\begin{proof}
    Let $q$ be a "$\aC$-hom-closed" query $q$ that admits a "duplicable singleton support" $S$. If $S$ isn't minimal, then $q$ is the trivial query $\top$, thus $\FGMC{q}$ and $\Shapley{q}$ are both obviously tractable. Otherwise, since $S$ contains a constant outside of $C$, it only remains to show that it is an "island support". Let $S'$ be a set of facts with $\const(S)\cap\Const(S')\subseteq C$ and $S^*$ a minimal support for $q$ in $S\cup S'$. If the single element of $S$ is contained in $S^*$, then $S^*=S$ by minimality, otherwise $S^*\subseteq S'$. We have thus shown $q$ to be pseudo-connected, and it suffices to apply \Cref{lem:pseudoconn}. 
\end{proof}

As for the known dichotomy for $\Shapley{}$ over "sjf-CQs" \cite{livshitsShapleyValueTuples2021}, it turns out that, although not all "sjf-CQs" are "pseudo-connected", the intractable queries always have an intractable component which is ``connected via "variables"''. This motivates the following definition and variant of \Cref{lem:pseudoconn} above.

\AP
A "Boolean" "CQ" $q$ is ""variable-connected"" if the "incidence graph" $\incgraph{\atoms(q)}$ remains connected after removal of the nodes $\const(q)$.
%
We say that 
a "$\aC$-hom-closed" query is \reintro{variable-connected} if it is equivalent to an infinitary union of "variable-connected" "$\aC$-hom-closed" "Boolean" "CQs".
Observe that any "hom-closed" query is "connected" if{f} it is "variable-connected".


\AP
A "fact" $\alpha$ is called a $q$-""leak"" if there exists some "fact" $\alpha'$ from some "minimal support" of $q$ and a "$\aC$-homomorphism" $h: \set{\alpha'} \to \set{\alpha}$ such that $h(c) \in \aC$ for some $c \in \const(\alpha') \setminus \aC$.  Intuitively, a $q$-"leak" is a structure that allows a "minimal support" of a "variable-connected" $q$ to intersect two "databases" that share no "constant" outside of $C$, by instantiating a "variable" with a "constant" in $C$.
For instance, consider the ($\{a\}$-"hom-closed") "CRPQ" $q\defeq \exists x ~ [AB+BA](x,a) \equiv \exists x,y ~  (A(x,y)\land B(y,a)) \lor (B(x,y)\land A(y,a))$. The "fact" $A(b,a)$ is a $q$-"leak" because of the first "fact" of the "minimal support" $\set{A(b,d), B(d,a)}$ and the mapping $\set{b \mapsto b, d \mapsto a}$.

\begin{lemma}\label{lem:leak}
Let $q, q'$ be "$C$-hom-closed@$C$-hom-closed" and "$C'$-hom-closed@$C$-hom-closed" queries, respectively, such that
\begin{enumerate}
    \item\label{leak:1} $q$ is "variable-connected";
    \item there is a "minimal support" $S'$ of $q'$ such that
    \begin{enumerate}
        \item\label{leak:2a} $S' \not\models q$,
        \item\label{leak:2b} $S'$ has no $q$-"leak",
        \item\label{leak:2c} if a "fact" $\alpha \in S'$ is "relevant to" $q$, then $\const(\alpha) \not\subseteq \aC$;
    \end{enumerate}
    \item\label{leak:3} $q$ has a "minimal support" $S$ with no $q$-"leak".
\end{enumerate}
Then, $\FGMC{q} \polyrx \Shapley{q \land q'}$.
\end{lemma}

\begin{corollary}\label{cor:non-h-cq}
    Let $q$ be a "non-hierarchical" "sjf-CQ" or a "non-hierarchical" constant-free "CQ". Then, $\Shapley{q}$ is "shP"-hard.
\end{corollary}
\begin{proof}
    The fact that the "CQ" $q$ is "non-hierarchical" can be seen in a "variable-connected" subquery. In other words, $q$ can be written as $q\land q'$ with $q$ "variable-connected" and "non-hierarchical". Now build for $i\in [2]$ a minimal support $S_i$ of $q_i$ by isomorphically mapping every variable to a fresh constant.
    
    The key idea is that there cannot be any "leak" if the "CQ" is either constant-free or "self-join-free". This is because a "leak" implies a "$\aC$-homomorphism" from an atom of the query to an atom of the "support" where a constant outside of $\aC$ is mapped to a constant in $\aC$. Now if $\aC=\emptyset$ this is obviously impossible, and if $q$ is "self-join-free", then for every atom $\alpha\in S_i$, the only atom in $q$ for which there exists a homomorphism is the one it is isomorphic to.
    
    At this point the hypotheses (\ref{leak:1}), (\ref{leak:2b}) and (\ref{leak:3}) are verified. Now (\ref{leak:2a}) is free since $q'$ can be removed if it is redundant, and (\ref{leak:2c}) must also be true because $\const(\alpha)\subseteq C$ is impossible if the query is constant-free and if it is "self-join-free" no fact in a minimal support of $q'$ can be "relevant to" $q$ because of the disjoint vocabularies.
    
    We can thus apply \Cref{lem:leak} to obtain $\FGMC{q}\polyrx \Shapley{q}$, and since $q$ is "non-hierarchical", this means that $\Shapley{q}$ is "shP"-hard because $\GMC{q}$ is \cite{kenigDichotomyGeneralizedModel2021}.
\end{proof}

The preceding lemma allows us, in particular, to recapture the dichotomy over "sjf-CQs" established in \cite{livshitsShapleyValueTuples2021}. It does not prove a full dichotomy over constant-free "CQ"s however because there could potentially be "hierarchical" queries which are "shP"-hard.


\subsection{Disconnected Queries}

Another approach for reducing $\FGMC{}$ to $\Shapley{}$ is to have a query whose answers on any database can be computed as two queries over disjoint subsets of the data. We can in fact view the "database" as being "probabilistic@probabilistic database" due to the $\FGMC{} \polyeq \SPPQE{}$ equivalence. In this setting, the intuition is that we compute the probability of each subquery by replacing the part of the query that does not affect it by some minimal support $S$ of the other subquery, apply the same construction as before to compute the probability of this subquery, and finally multiply the two probabilities together.

\AP
Formally, we say that a "$\aC$-hom-closed" query $q$ is ""decomposable into"" $q_1 \land q_2$ if it is equivalent to $q_1 \land q_2$, where $q_1,q_2$ are "Boolean" queries such that
\begin{enumerate}
    \item there are "minimal supports" $S_1, S_2$ of $q_1, q_2$ respectively,  with $\const(S_1) \not\subseteq \aC$ and $\const(S_2) \not\subseteq \aC$,
    \item for all "minimal supports" $S_1, S_2$ of $q_1, q_2$ respectively, we have $S_1 \cap S_2 = \emptyset$.
\end{enumerate}

\begin{lemma}\label{lem:decomposable}
    For every "decomposable" "$\aC$-hom-closed" query~$q$, we have
    $\FGMC{q} \polyrx \Shapley{q}$.
\end{lemma}

The most obviously "decomposable" queries are those that can be written as the conjunction of two independent subqueries over distinct vocabularies.
In fact, in the constant-free case, this syntactic class covers all "decomposable" queries. 

\begin{lemmarep}\label{lem:chara-decomp} A 
    "hom-closed" "Boolean query" $q$ is "decomposable" iff 
 it is equivalent to $q_1\land q_2$, where $q_1$ and $q_2$ are "Boolean" "UCQ$^\infty$" queries that do not share any "relation name". 
\end{lemmarep}
\begin{proof}
    First observe that the first condition of "decomposability" is always satisfied when the query is "hom-closed". It will thus be ignored from now.\\
    
    \noindent$(\Rightarrow)$\quad Assume $q$ is "decomposable into" $q_1\land q_2$. We may further assume $q_1$ and $q_2$ are "UCQ$^\infty$"s since any "hom-closed" query can be written as such. Let us now assume by contradiction that $q_1$ and $q_2$ share a "relation name" $P$. We can assume w.l.o.g.\ that they admit respective minimal supports $S_1$ and $S_2$ that contain the "relation name" $P$: if $q_i$ does not admit any minimal support containing $P$, then $P$ can be removed from any conjunct of $q_i$ it appears in without any consequence. By homomorphically renaming their constants, we can assume that both $S_1$ and $S_2$ contain the atom $P(\vec{a})$ that only contains the fresh constant $a$, which contradicts $S_1\cap S_2=\emptyset$.\\
    
    \noindent$(\Leftarrow)$\quad Assume $q$ can be written as $q_1\land q_2$ where $q_1$ and $q_2$ are "UCQ$^\infty$"s that do not share any variable nor any "relation name". Let $S_1$ and $S_2$ be respective supports of $q_1$ and $q_2$. Since  they do not share any "relation name" then $S_1\cap S_2 =\emptyset$.
\end{proof}
This characterization does not apply to "$C$-hom-closed" queries,
as witnessed by the query
 $\exists x,y\, R(a,x) \land R(b,y)$ which is "decomposable into" $(\exists x\, R(a,x)) \land (\exists y \, R(b,y))$ but admits no disjoint-vocabulary decomposition. 

\AP
\Cref{lem:decomposable} allows us to establish a dichotomy for the constant-free ""sjf-CRPQ"", that is "CRPQs" whose "path atoms" have pairwise distinct vocabularies, and more generally constant-free ""cc-disjoint-CRPQ"", that is constant-free "CRPQs" whose connected components are over pairwise disjoint vocabularies. Of course the connected cases are not "decomposable", but they are "pseudo-connected" hence \Cref{lem:pseudoconn} applies instead.\siderev{New footnote observing that there is a query preserving reduction for constant-free "sjf-CQ"s. Also modified \Cref{fig:query-lang} to reflect this fact.}
\footnote{\changed{In fact this argument also applies to constant-free "sjf-CQ"s, yielding a query-preserving reduction, whereas \Cref{lem:leak} uses a different query for the reduction.
}}
\changed{Note that the dichotomy is effective, since the boundedness of "CRPQ"s is decidable \cite{barcelo_et_al:LIPIcs.ICALP.2019.104} and so is determining whether a "UCQ" is safe \cite{dalviDichotomyProbabilisticInference2012}.}\siderev{Added note on effectiveness.}

\begin{corollary}
    Let $q$ be a constant-free "cc-disjoint-CRPQ". Then $\Shapley{q}$ is in "FP" if it can be written as a "safe" "UCQ", or "shP"-hard otherwise.
\end{corollary}
\begin{proof}
    If $q$ is "connected" then it is "pseudo-connected" by \Cref{lem:connected-is-pseudoconnected}, otherwise it is "decomposable" by \Cref{lem:chara-decomp}. Either way, $\FGMC{q}\polyrx \Shapley{q}$ by \Cref{lem:pseudoconn} or \Cref{lem:decomposable}. Then, since $q$ is a "hom-closed" "graph query", it is in "FP" if it can be written as a "safe" "UCQ", or "shP"-hard otherwise \cite{amarilliUniformReliabilityUnbounded2023a,dalviDichotomyProbabilisticInference2012,kenigDichotomyGeneralizedModel2021} (see the proof of \Cref{cor:conn-hom-closed}).
\end{proof} 
\section{Proofs of main results}
\label{sec:mainproofs}
Our main results are general-purpose reductions from $\FGMC{q}$ to $\Shapley{q}$ for "$C$-hom-closed" queries $q$ verifying certain properties. Let us first give a general idea of the reductions we will use.

Given a "partitioned database" $\D=\Dn \dcup \Dx$ and a query $q$, our goal is to compute $\FGMC{q}$ on $\D$ in polynomial time with the help of an oracle for $\Shapley q$. 
The basic idea for doing this is to add a "minimal support" $S$ of $q$, and make every "fact" "exogenous" except for one, denoted by $\mu$. That is, $S = (\Dn[S], \Dx[S])$ where $\Dn[S] = \set \mu$.
If the query has the right properties, the arrival of $\mu$ will have no impact on the satisfaction of $q$ if, and only if, the set of players before $\mu$ forms a "generalized support for $q$ in" $\D$. 
In such case, the Shapley value $\Sh(\Dn \dcup  \Dn[S], \scorefun ,\mu)$ on $(\Dn 
\dcup \Dn[S], \Dx \dcup \Dx[S])$ will be an affine combination of the different counts of $\FGMC{q}$ on $\D$, and by building variants we obtain a system of linear equations that can be inverted to compute the values in $\FGMC q$ on $\D$.

However, for this technique to be successful, one must ensure that no "minimal support" of $q$ can intersect both $\D$ and $S$, otherwise there could be a subset of $\Dn \cup \Dn[S]$ which is a "generalized support for $q$ in" $(\Dn \dcup \Dn[S], \Dx \dcup \Dx[S])$, but not in $\D$.

\siderev{Restructured subsections: one subsection per key lemma.}

\changed{We will first show in \Cref{ssec:main-proof} the reduction of \Cref{lem:leak}, which is the most involved 
as different queries are used for the two problems ($\FGMC{q}$ and $\Shapley{q \land q'}$). 
We then show in \Cref{ssec:proof-pseudoconn} that a very similar construction works when assuming the hypotheses of \Cref{lem:pseudoconn} instead
. We finally show the reduction of \Cref{lem:pseudoconn} in \Cref{ssec:lem:decompoable-proof}, which is somewhat different in spirit.}
\subsection{Proof for \Cref{lem:leak}}\label{ssec:main-proof}
\changed{As a first step, we observe that we can restrict our attention to databases with certain properties:}
\begin{claimrep}\label{claim:wlog}
	\siderev{Turned ``"wlog" assumptions'' into formal statements, as suggested.}
	Let $q$, $q'$, $\aC$, $\aC'$, $S$ and $S'$ satisfy the hypotheses of \Cref{lem:leak}. 
	If there exists a polynomial-time algorithm using unit-cost calls to $\Shapley{q \land q'}$ that computes $\FGMC{q}$ under the assumption that the input database $\D$ satisfies:
	
	\begin{enumerate}
		\item\label{cas5.1} $\Dx \not\models q$,
		\item\label{cas5.2} $\const(\D) \cap \const(S')\subseteq \aC$, and
		\item\label{cas5.3} $\D\cap S' = \emptyset$,
	\end{enumerate}
	then $\FGMC{q} \polyrx \Shapley{q \land q'}$.
\end{claimrep}
\begin{proofsketch}
	Take an arbitrary $\D$. If $\Dx \models q$, then every subset of $\Dn$ is a "generalized support" for $q$, hence counting is trivial. Otherwise, we can $\aC$-isomorphically rename $\D$ to satisfy (\ref{cas5.2}) and the remaining facts in $\D\cap S'$ are "irrelevant to" $q$ because of \Cref{lem:leak}.\eqref{leak:2c} so they can be removed from $\D$ without consequence.
\end{proofsketch}
\begin{appendixproof}
	Let $q, q'$ satisfy the hypotheses of \Cref{lem:leak}, and let $\D$ be any "partitioned database" on which we wish to compute $\FGMC{q}$. If $\Dx \models q$, then every subset of $\Dn$ is a "generalized support" for $q$, therefore counting is trivial, hence \eqref{cas5.1}.
	
	Let $S'$ be the "minimal support" of $q'$ given by the hypotheses. We $C$-isomorphically rename $\D$ so that all constants in $\const(\D) \cap \const(S') \setminus \aC$ are replaced by fresh ones. Since $q$ is $\aC$-hom-closed, this new database is no identical as far as $\FGMC{q}$ is concerned, and it satisfies \eqref{cas5.2}.
	
	At this stage, all "facts" of  $\D\cap  S'$ are "irrelevant to" $q$ due to \Cref{lem:leak}.\eqref{leak:2c}. To obtain \Cref{claim:wlog}.\eqref{cas5.3} it only remains to show that such facts can be ignored without impacting $\FGMC{q}$. To see why, let $\alpha \in \D\cap  S'$. Then from $\FGMC{q}(\D\setminus\{\alpha\})$, one may use \Cref{prop:reductions-problems}.\eqref{cas1.1} to compute $\SPPQE{q}((\D\setminus\{\alpha\})^p)$ for any $p\in (0,1)$, where $(\D\setminus\{\alpha\})^p$ denotes the probabilistic version of $\D\setminus\{\alpha\}$ where all "endogenous" (resp.\ "exogenous") "facts" have probability $p$ (resp.\ probability $1$). 
	Since $\alpha$ is "irrelevant to" $q$, $\SPPQE{q}((\D\setminus\{\alpha\})^p)=\SPPQE{q}(\D^p)$, 
	and from there we can use again the equivalence of \Cref{prop:reductions-problems}.\eqref{cas1.1} to obtain $\FGMC{q}(\D)$.
\end{appendixproof}


For the rest of the proof we shall fix some $\D$ which satisfies \Cref{claim:wlog}. To build our reduction, we shall first ``complete'' $\D$ into $\D'$ by adding facts in such a way that $\FGMC{q}(\D)=\FGMC{q \land q'}(\D')$. Then we shall ``duplicate'' part of a "minimal support" of $q$ to compute $\FGMC{q \land q'}(\D')$ using an
$\Shapley{q \land q'}$ oracle.

\begin{claimrep}[\textit{completion}]\label{claim:completion}
	\siderev{Turned ``"wlog" assumptions'' into formal statements, as suggested.}
	Take $\D' \defeq \D \uplus  S'$, where $\Dx[\D'] = \Dx \dcup S'$ ("ie", all "facts" of $S'$ are "exogenous"). Then for every input size $j$, 
	$\FGMC{q}(\D,j)=\FGMC{q \land q'}(\D',j)$.
\end{claimrep} 
\begin{proofsketch}
	We observe that $\D$ and $\D'$ have the same "endogenous" "facts", then show that the "generalized supports" for $q$ in $\D$ and those for $q \land q'$ in $\D'$ are precisely the same sets.
\end{proofsketch}
\begin{appendixproof}
We assume $\D$ satisfies \Cref{claim:wlog}. Observe that $\D$ and $\D'$ have the same "endogenous" "facts". We now show that the "generalized support" for $q$ in $\D$ and those for $q \land q'$ in $\D'$ are precisely the same sets.

For the left-to-right inclusion, let $B \subseteq \Dn$ be a "generalized support" for $q$ in $\D$ ("ie", $B\cup\Dx\models q$). Since $S'\models q'$ and $q'$ is "$C'$-hom-closed", it follows that $B\cup\Dx[\D'] = B \cup \Dx \cup  S'\models q \land q'$.

To show the right-to-left inclusion, suppose towards a contradiction that we have a set $B \subseteq \Dn$ such that (i) $B\cup \Dx[\D']\models q \land q'$, and (ii) $B\cup\Dx\not\models q$. In particular, (i) implies that there exists a "minimal support" $S$ of $q$ in $B\cup\Dx[\D']$. 
Since we have assumed that $B\cup\Dx\not\models q$, and $S'\not\models q$ by \Cref{lem:leak}.\eqref{leak:2a}, this means that $S$ must intersect both $\D$ and $S'$. We can thus pick "facts" $\alpha \in S\cap\D$ and $\alpha' \in S\cap  S'$.  Since $q$ is "variable-connected" (\Cref{lem:leak}.\eqref{leak:1}), $S$ is the $\aC$-"homomorphic image@$C$-homomorphism" of some "minimal support" $S^\dagger$ in which every atom is connected to every other via constants outside of $\aC$.
In other words, there is a path $P$ in the "incidence graph" $\incgraph{S}$ linking $\alpha$ and $\alpha'$ only via "atoms" and "constants" that are the homomorphic image of a "constant" from $\const(S^\dagger)\setminus \aC$. However, since the only vertices that $\incgraph{\D}$ and $\incgraph{S'}$ have in common are the "constants" in $\aC$, the path $P$ must go through one of them, which contradicts \Cref{lem:leak}.\eqref{leak:2b} because it would witness a $q$-"leak" in $S'$.
\end{appendixproof}

\begin{claimrep}[\textit{duplication}]\label{claim:duplication}
\siderev{Turned ``"wlog" assumptions'' into formal statements, as suggested.}
	There exists a family $(S^k)_{k\in\IN}$ of distinct $\aC$-isomorphic sets of facts, and a set $S^-$ disjoint with every $S^k$ such that every $S^k\uplus S^-$ is a $q$-"leak"-free "support" of $q$ in which every atom is connected to every other by some constant outside of $\aC$, and that shares no constant with $\D'$ or $\aC\cup \aC'$, except for those in $\aC$.
\end{claimrep}
\begin{proofsketch}
	We can take $S$, the $q$-"leak"-free "support" of $q$ given by \Cref{lem:leak}.\eqref{leak:3}, tweak it a little using \Cref{lem:leak}.\eqref{leak:1} to avoid shared constants, then partition it into $S^0\uplus S^-$ such that $S^0$ contains at least one constant $a$ outside of $\aC$ to allow for infinitely many $\aC$-isomorphic copies.
	%
\end{proofsketch}
\begin{appendixproof}
Take $S$ to be the $q$-"leak"-free "support" of $q$ given by \Cref{lem:leak}.\eqref{leak:3}. Since $q$ is "variable-connected" by \Cref{lem:leak}.\eqref{leak:1}, $S$ is the $\aC$-homomorphic image of some "minimal support" $S^\dagger$ in which every atom is connected to every other by some constant outside of $\aC$. There cannot be any $q$-"leak" in $S^\dagger$ either, otherwise the homomorphism would transfer it to $S$. By replacing $S$ with $S^\dagger$ we can therefore assume "wlog" that it satisfies this extra property, and by $\aC$-isomorphically renaming it we can further assume that it shares no constant with $\D'$ or $\aC\cup \aC'$, except for those in $\aC$.

Finally, we take any constant $a \in \const(S) \setminus (C \cup C')$ and partition $S$ into the set $S^0$ of facts that contain $a$ and the set $S^-$ of facts that do not. We then build the rest of the desired family $(S^k)_{k\in\IN}$ of distinct $\aC$-isomorphic copies of $S^0$ by replacing $a$ with fresh constants. %
\end{appendixproof}

Now that we have built $\D'$, $(S^k)_{k\in\IN}$ and $S^-$ in \Cref{claim:completion,claim:duplication}, we are ready to describe the reduction. Fix a distinguished fact $\mu$ of $S^0$ 
and for every $k$ denote by $\mu^{k}$ the corresponding fact in~$S^{k}$. 
\AP
Consider the database
\[\intro*\Ai\defeq \D' \cup S^0\cup S^{1}\cup\dots\cup S^{i}\cup S^-\]
depicted in \Cref{fig:redx}.\siderev{Improved figure and caption for clarity.} Now we partition $\Ai$ into "endogenous" and "exogenous" "facts". The "endogenous" facts consist of: 
\begin{itemize}
	\item the "endogenous" facts in $\D'$, "ie" $\Dn$ (recall that $S' \subseteq \Dx[\D']$)
	\item the "fact" $\mu$ and each of its copies $\mu^k$, 
	\item all "facts" in $S^-$.
\end{itemize}
All remaining "facts" of $\Ai$ are "exogenous". 

Note that we deviate slightly from the proof idea by only duplicating a part of S. This is so that, aside from $S'$, the only "exogenous" "facts" that we add in our construction are those in $S^0\setminus\{\mu\}$ and their copies. In particular none is added if $S'=\emptyset$ and $S^0$ is a singleton. This will prove important later in \Cref{sec:endogenous} when we consider purely endogenous databases.

\begin{figure}[t]
	\centering
	\definecolor{darkred}{rgb}{.85, 0, 0}
\definecolor{darkblue}{rgb}{0, 0, .9}

\begin{tikzpicture}{scale=.8}
	\begin{pgfonlayer}{nodelayer}
		\node [color=darkred,draw, circle] (0) at (0.75, 0.5) {\small$a$};
		\node [color=darkblue,draw, circle] (2) at (0.75, -0.5) {\small\phantom{$a$}};
		\node [] (3) at (0.75, 0) {\dots};
		\node [color=darkred] (4) at (2.5, 0.4) {$\times$};
		\node [color=darkred] (5) at (2.5, -0.4) {$\times$};
		\node [] (6) at (2, 0.75) {};
		\node [] (7) at (2, -0.75) {};
		\node [] (8) at (4, -0.75) {};
		\node [] (9) at (4, 0.75) {};
		\node [color=darkred] (10) at (3.25, 0) {$S^-$};
		\node [] (11) at (-0.25, 1) {};
		\node [] (12) at (5, 1) {};
		\node [] (13) at (4, 1) {};
		\node [] (14) at (-1.5, 0.5) {};
		\node [] (15) at (-1.5, 1.25) {};
		\node [] (16) at (-3, 1.25) {};
		\node [] (17) at (-3, 0.5) {};
		\node [] (18) at (-1.5, -0.75) {};
		\node [] (19) at (-1.5, 0) {};
		\node [] (20) at (-3, 0) {};
		\node [] (21) at (-3, -0.75) {};
		\node [] (22) at (-3.25, -0.75) {};
		\node [] (23) at (-3.25, 1.25) {};
		\node [] (24) at (-2.25, 0.875) {$\D$};
		\node [] (25) at (-2.25, -0.375) {$S'$};
		\node [] (26) at (-.15, 0.85) {};
		\node [] (27) at (-.15, 0.15) {};
		\node [] (28) at (-.15, -0.15) {};
		\node [] (29) at (-.15, -0.85) {};
	\end{pgfonlayer}
	\begin{pgfonlayer}{edgelayer}
		\draw [color=darkred,dashed] (6.center) to (7.center);
		\draw [color=darkred,dashed] (7.center) to (8.center);
		\draw [color=darkred,dashed] (8.center) to (9.center);
		\draw [color=darkred,dashed] (9.center) to (6.center);
		\draw [color=darkred, very thick,->, >= stealth, in=-135, out=135, loop] (0) to ();
		\draw [color=darkred, very thick,<-, >= stealth] (0) to (5.center);
		\draw [color=darkred,->, >= stealth] (0) to node[near start,above=-.1] {$\mu$} (4.center);
		\draw [color=darkred,decorate,decoration={brace,amplitude=5pt,raise=0.5pt}] (11.center) to node[midway,above=.1] {$S$} (13.center);
		\draw [color=darkred,decorate,decoration={brace,amplitude=5pt,raise=0.5pt}] (27.center) to node[midway,left=.1] {$S^0$} (26.center);
		\draw [color=darkblue,decorate,decoration={brace,amplitude=5pt,raise=0.5pt}] (29.center) to node[midway,left=.1] {$S^i$} (28.center);
		\draw [decorate,decoration={brace,amplitude=5pt,raise=0.5pt}] (22.center) to node[midway,left=.1] {$\D’$} (23.center);
		\draw [color=darkblue, very thick,->, >= stealth, in=-135, out=135, loop] (2) to ();
		\draw [color=darkblue, very thick,<-, >= stealth] (2) to (5.center);
		\draw [color=darkblue,->, >= stealth] (2) to node[very near start,above=-.1] {$\mu^i$} (4.center);
		\draw [dashed, very thick] (21.center) to (18.center);
		\draw [dashed, very thick] (18.center) to (19.center);
		\draw [dashed, very thick] (19.center) to (20.center);
		\draw [dashed, very thick] (20.center) to (21.center);
		\draw [dashed] (17.center) to (14.center);
		\draw [dashed] (14.center) to (15.center);
		\draw [dashed] (15.center) to (16.center);
		\draw [dashed] (16.center) to (17.center);
	\end{pgfonlayer}
\end{tikzpicture}
	\caption{\changed{Illustration of the construction of $\Ai$ on a "graph database", where "constants" and "facts" are depicted by vertices and arrows, respectively.
	Graphically disconnected parts do not share any constant except for those that appear in $\aC$. New "exogenous" facts are thick arrows.}}
	\label{fig:redx}
\end{figure}
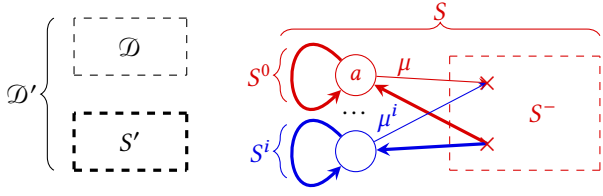

We now use a $\Shapley{q \land q'}$ oracle to compute the Shapley value of $\mu$ in the "partitioned database" 
$\Ai$. This value is given by the formula:
\[ 
\Sh(\Ain,\scorefun[],\mu) = \sum_{B\subseteq \Ain\setminus\{\mu\}} \frac{|B|!(|\Ain|-|B|-1)!}{|\Ain|!} (\scorefun[](B\cup\{\mu\}) - \scorefun[](B)),
\]
where we abbreviate $\scorefun[q\land q']$ as $\scorefun[]$.

We aim to express $\mathrm{Sh}(\Ain,\scorefun[],\mu)$ as a linear combination (plus constant) of the different numbers $GSMC_j$ of "generalized supports for" $q \land q'$ in $\D'$ of size $j$, "ie" $\FGMC{q\land q'}(\D',j)$ which by \Cref{claim:completion} equals $\FGMC{q}(\D,j)$, the final value we wish to compute. To do so, we observe that our construction satisfies the following lemma.

\begin{restatable}{lemma}{lempcnx}\label{lem:pcnx}
	$(\scorefun[](B\cup\{\mu\}) - \scorefun[](B))=0$ 
	iff one of the following three mutually exclusive cases holds:
	\begin{enumerate}
		\item\label{cas2.1} $\exists k\in [i]. \mu^{k} \in B$,
		\item\label{cas2.2} $\forall k\in [i]. \mu^{k} \notin B
		~~~\land~~~
		\exists \alpha\in S^-. ~ \alpha \notin B$,
		\item\label{cas2.3} 
		$\forall k \in [i]. \mu^{k} \notin B
		~~~\land~~~
		\forall \alpha \in S^-. ~ \alpha\in B
		~~~\land~~~(B\cap\Dn) \cup \Dx \models q$.
	\end{enumerate}
\end{restatable}

\begin{proof}
	Case (\ref{cas2.1}) is when there is already a copied fact in $B$ that makes $\mu$ irrelevant. Case (\ref{cas2.2}) is when, although $\mu$ has not been made redundant, there is some other fact of $S^-$ missing to satisfy $q$. Case (\ref{cas2.3}) is when the arrival of $\mu$ could have made an impact but $q$ was already satisfied by $(B\cap\Dn) \cup \Dx$. Recall that we only need to consider $q$ because the $q'$ component is always satisfied by the "exogenous" $S'$. In all other cases: (i) $B\cup\{\mu\}\cup\Aix\models q\land q'$ 
	because the presence of $S^-\cup\{\mu\}$ suffices to satisfy $q$, and (ii) $B\cup\Aix \not \models q$ (hence $B\cup\Aix \nvDash q\land q'$) 
	because of the following observations. Assume for a contradiction that there is a "minimal support" of $q$ in $B\cup\Aix$. It cannot be $S$ because $\mu\notin B$, nor one of its copies otherwise case (\ref{cas2.1}) would be met, nor can it be contained in $  S'$ because $  S'\not\models q$ by \Cref{lem:leak}.\eqref{leak:2a}, nor in $\D$ otherwise case (\ref{cas2.3}) would be met. The only remaining possibility is that the "minimal support" intersects several of the three parts, but since they have no constant in common outside of $\aC$ and $q$ is "variable-connected", this witnesses a $q$-"leak", contradicting \Cref{lem:leak}.\eqref{leak:2b} or \Cref{claim:duplication}.
\end{proof}

We finally show that \Cref{lem:pcnx} suffices to conclude. 
Consider the complement of the formula above, that is:
\[ 
1-\Sh(\Ain,\scorefun[],\mu) = \sum_{\substack{B\subseteq \Ain\setminus\{\mu\}, \\ B \text{ verifies (\ref{cas2.1}), (\ref{cas2.2}) or (\ref{cas2.3})}}} \frac{|B|!(|\Ain|-|B|-1)!}{|\Ain|!}
\]
It suffices to enumerate these cases since in all others $(\scorefun[](B\cup\{\mu\}) - \scorefun[](B))=1$ by \Cref{lem:pcnx}. \siderev{Improved formulation for clarity.}
\changed{Moreover, the contribution $Z$ of the summation terms corresponding to subsets $B$ satisfying case (\ref{cas2.1}) or (\ref{cas2.2}) can be easily computed in polynomial time, and then subtracted out.}
We are thus left with the problem of computing
\[ 
	1-\Sh(\Ain,\scorefun[],\mu) - Z = Sh^{i} \defeq \hspace{-1em} \sum_{B\subseteq \Ain\setminus\{\mu\}, B \text{ verifies (\ref{cas2.3})}} \hspace{-1em} \frac{|B|!(|\Ain|-|B|-1)!}{|\Ain|!}.
\]
Since the $B$ sets in the summation do not contain $\mu$ nor any $\mu^{k}$ and contain $S^-$, they are exactly the sets of the form $G\cup S^-$ where $G$ is a "generalized support for" $q$ in $\D$. Grouping the $G$'s by size and simplifying, one obtains the following formula.

\begin{align*}
	Sh^{i} 
			&= \sum_{j=0}^{|\Dn|} \frac{(j+|S^-|)!(|\Dn|+i-j)!}{(|\Dn|+i+|S^-|+1)!}\cdot GSMC_j
\end{align*}

Finally, by performing this computation for $i$ ranging from 0 to $|\Dn|$, we get a system of equations linking the $Sh^{i}$ to the $GSMC_j$ whose underlying matrix can be reduced (multiplying every line by $(|\Dn|+i+|S^-|+1)!$, dividing every column by $(j+|S^-|)!$ then reversing the column order) to the matrix of general term $(i+j)!$ which by \cite[proof of Theorem 1.1]{bacherDeterminantsMatricesRelated2002} is invertible.

\subsection{Proof of \Cref{lem:pseudoconn}}
\label{ssec:proof-pseudoconn}
This proof\siderev{New subsection to discuss adaptation for Lemma 4.1} is very similar to the one for \Cref{lem:leak}, in that we build a similar instance and show it satisfies \Cref{lem:pcnx}, which suffices to conclude.
Letting  $q_{pc}$ be the input pseudo-connected query, we shall follow the construction from \Cref{ssec:main-proof}, but using 
$q\defeq q_{pc}$ (and $\aC_{pc}$ for $\aC$), $q'\defeq\top$ ("ie", the always-true query), and letting $S$ be an "island" "minimal support" for $q_{pc}$ 
containing some $a\in \Const(S)\setminus\aC$. An isomorphic renaming ensures that no constant is shared between $S$ and $\D$ besides those in $C$. 
We then observe that if any of the conditions of \Cref{lem:pcnx} is met, then $\mu$ is redundant for the exact same reason as before, and otherwise the only way of having $(\scorefun[](B\cup\{\mu\}) - \scorefun[](B))=0$ is to have a minimal support of $q$ that intersects both $B$ and $S$ or one of its copies. However, this is precisely forbidden by the fact that $S$ is an "island support" for $q$.

\subsection{Proof for \Cref{lem:decomposable}}
\label{ssec:lem:decompoable-proof}

\AP
Let $q$ be "decomposable into" $q_1\land q_2$ and $\D$ be an input "partitioned database". Recall that $\FGMC{q}\polyeq\SPPQE{q}$ by \Cref{prop:reductions-problems}. 
By definition, no fact can be "relevant to" both $q_1$ and $q_2$, hence $\D$ can be partitioned into $\D^1\uplus\D^2$ such that no fact of $\D^1$ (resp.\ $\D^2$) is "relevant to" $q_{2}$ (resp.\ $q_1$). 
Therefore, the $\SPPQE{q}$ problem over $\D$ can be solved by computing separately the $\SPPQE{q_i}$ over $\D_i$, then multiplying the probabilities together. Using \Cref{prop:reductions-problems} in the other direction, we can obtain the latter values by computing $\FGMC{q_1}$ (resp.\ $\FGMC{q_2}$) over $\D_1$ (resp.\ $\D_2$). Due to symmetry, it suffices to show how to compute $\FGMC{q_1}$ over $\D_1$. 

To this end, we build the same reduction as in \Cref{ssec:main-proof} with $S'\defeq \emptyset$ except that we use any "minimal support" of $q_2$ as $S$. The reduction will be successful because $\mu$ can only contribute if $B\cup\Dx{}\models q_1$, and the facts in $S'$ are "irrelevant to" $q_1$ since they are "relevant to" $q_2$. The detailed construction goes as follows.

Let $S$ be a "minimal support" of $q_2$ that contains at least one constant $a$ outside of $C$. Similarly to what has been done in the proof of \Cref{lem:leak,lem:pseudoconn}, we build the parametrized database $A^{i}\defeq \D\cup S^0\cup S^{1}\cup\dots\cup S^{i}\cup S^-$ from the subset $ S^0\subseteq  S$ of facts that contain $a$, and use the oracle to compute the same Shapley value of some distinguished $\mu\in S^0$, given by the same formula:
\[ 
\Sh(A^{i}_{\mathsf{n}},\scorefun,\mu) = \sum_{B\subseteq A^{i}_{\mathsf{n}}\setminus\{\mu\}} \frac{|B|!(|A^{i}_{\mathsf{n}}|-|B|-1)!}{|A^{i}_{\mathsf{n}}|!} (\scorefun(B\cup\{\mu\}) - \scorefun(B)) 
\]

Now we express $\mathrm{Sh}(A^{i}_{\mathsf{n}},\scorefun,\mu)$ as a linear combination (plus constant) of the different numbers $FGMC_j$ of "generalized supports for" $q_1$ in $\D^1$ with size $j$. For simplicity we will rather express it with their complement $\overline{FGMC_j}$ (number of subsets of size $j$ that are \textit{not} "generalized supports") from which we can as easily obtain the value $\FGMC{q_1}(\D^1)$ we want. Indeed, we can observe that, similarly to 
 \Cref{lem:pcnx}, $(\scorefun(B\cup\{\mu\}) - \scorefun(B))=0$ in the following three mutually exclusive cases:
\begin{enumerate}
	\item\label{cas3.1} $\exists k\in [i]. \mu^{k} \in B$
	\item\label{cas3.2} $\forall k\in [i]. \mu^{k} \notin B
	~~~\land~~~
	\exists \alpha\in S^-. \alpha \notin B$
	\item\label{cas3.3} $\forall k \in [i] ~ \mu^{k} \notin B
	~~~\land~~~
	\forall \alpha \in S^- ~ \alpha\in B
	~~~\land~~~ (B\cap\Dn[\D^{1}]) \cup \Dx[\D^{1}] \nvDash q_1 
	$
\end{enumerate}
In the first case, $B$ contains a fact $\mu^{k}$ which completes the corresponding $ S^{k}$ and makes $\mu$ irrelevant in satisfying its component $q_2$ of the query. In the second case, $S^-$ is incomplete, meaning $\mu$ cannot help to satisfy $q_2$. In the last case $\mu$ cannot contribute because the $q_1$ subquery isn't satisfied.

Note that in all other cases the addition of $\mu$ makes $q_2$ and $q$ satisfied where they previously weren't: $q$ is satisfied because of the respective models $(B\cap\Dn[\D^{1}]) \cup \Dx[\D^{1}]$ and $S$ of $q_1$ and $q_2$, and $q_2$ previously wasn't because $S$ is a  "minimal@support" and the facts of $\D^1$ are "irrelevant to" $q_2$. \\

By taking the complement and removing a constant term as before, we are left with:
\[ 
Sh^{i} \defeq \sum_{B\subseteq A^{i}_{\mathsf{n}}\setminus\{\mu\},\, B \text{ verifies (\ref{cas3.3})}} \frac{|B|!(|A^{i}_{\mathsf{n}}|-|B|-1)!}{|A^{i}_{\mathsf{n}}|!}.
\]
Since the $B$ sets do not contain $\mu$ nor any $\mu^{k}$, they are exactly the sets that are \emph{not} "generalized supports" of $q_1$ in $\D^1$. Grouping them by size, one obtains the following formula.
\[Sh^{i} = \sum_{j=0}^{|\Dn[\D^{1}]|} \frac{j!(|A^{i}_{\mathsf{n}}|-j-1)!}{|A^{i}_{\mathsf{n}}|!}\cdot \overline{FGMC_j}
= \sum_{j=0}^{|\Dn[\D^{1}]|} \frac{j!(|\Dn[\D^{1}]|+i-j)!}{(|\Dn[\D^{1}]|+i+1)!}\cdot \overline{FGMC_j}\]
%
\noindent We can 
conclude by the same 
linear algebra argument as in the proof of \Cref{lem:leak,lem:pseudoconn}. 
\section{Extensions and variants}
\label{sec:extensions}
This section explores 
variants of the considered problems, namely: 
computing the "Shapley value" on databases without "exogenous" "facts" (\S \ref{sec:endogenous}), 
extending the reductions to queries with negation (\S \ref{sec:negation}),
computing the maximum "Shapley value" (\S \ref{sec:maximum}), and
initiating the study of the "Shapley value" of "constants" rather than "facts" (\S \ref{sec:Shapley-constants}).
\subsection{Purely Endogenous Databases}
\label{sec:endogenous}
\begin{toappendix}
    \subsection{Purely Endogenous Databases}
    \label{app:endogenous}
\end{toappendix}
Having the possibility of fixing some "exogenous" facts makes the $\Shapley{}$ problem more general and flexible. 
However, often we will be presented with an unpartitioned database, 
in which case all "facts" should be treated as "endogenous" and taken as players in the Shapley game. 
Unfortunately, $\Shapley{}$ hardness results crucially rely on the existence of "exogenous" facts, and whether similar hardness results hold for
purely endogenous databases
is unclear. While we do not resolve this challenging question here, we make steps towards understanding
how $\Shapleyn q$, the restriction of $\Shapley q$ to  "partitioned databases" of the form $\D = (\Dn, \emptyset)$, 
is related to other problems. 


Let us start by observing that the proof of $\Shapley q \polyrx \FGMC q$ in \Cref{prop:reductions-problems} adds an "exogenous" "fact" to the input database, meaning it cannot prove $\Shapleyn q \polyrx \FMC q$ immediately. However, the following lemma shows that a constant number of "exogenous" "facts" is not a problem under "polynomial-time Turing-reductions".


\begin{lemmarep}\label{lem:FGMC-MC}
	Fix $k \geq 1$ and consider a query $q$ and instance $\D = (\Dn,\Dx)$ of $\FGMC q$ with $|\Dx|=k$. Then $\FGMC q$ on $\D$ can be solved by a polynomial-time algorithm that performs $2^k$ calls to an oracle for $\FMC q$.
\end{lemmarep}
\begin{appendixproof}
	Denote by $\FMC{q}^{(k)}$ the "fixed-size generalized model counting problem" restricted to inputs that contain at most $k$ "exogenous" "facts". For $k>0$, let $\D = (\Dn, \Dx)$ be an instance of $\FGMC q$, with $|\Dx|=k$. Take any "fact" $\alpha\in\Dx$. Then, the "generalized supports" of size $j$ in $\D$ are exactly the "generalized supports" of size $j+1$ in $(\Dn \cup \set \alpha,\Dx \setminus \set \alpha)$ that contain $\alpha$, which leads to the following formula:
	\begin{align*}
		\mathrm{FMC}^{(k)}_j(q)(\D_{\mathsf{n}},\D_{\mathsf{x}}) &=
		\mathrm{FMC}^{(k-1)}_{j+1}(q)(\D_{\mathsf{n}}\cup\{\alpha\},\D_{\mathsf{x}}\setminus\{\alpha\})\\
		&\hspace{3mm}-\mathrm{FMC}^{(k-1)}_{j+1}(q)(\D_{\mathsf{n}},\D_{\mathsf{x}}\setminus\{\alpha\})
	\end{align*}

	By recursively applying this formula we get the value we want after $2^k$ calls to $\FMC q$.
\end{appendixproof}


\begin{corollary}[of \Cref{lem:FGMC-MC} and the proof of \Cref{prop:reductions-problems}]\label{cor:shn-tract}
    For every query $q$, it holds that $\Shapleyn q \polyrx \FMC q$.
\end{corollary}

We now turn to the reductions in the other direction, useful for showing hardness. By adding a hypothesis,
we can adapt 
\Cref{lem:pseudoconn} and \Cref{lem:decomposable} 
to the purely endogenous setting: 

\begin{lemma}[Adaptation of \Cref{lem:pseudoconn}]\label{lem:endo-p-con}
    Let $q$ be a "$\aC$-hom-closed" query that has an "island" "minimal support" $S$ and a "constant" $a \notin C$ appearing in exactly one "fact" of $S$. Then $\FMC{q}\polyrx \Shapleyn{q}$.
\end{lemma}
\color{black}
Note for instance that for any "connected" minimal "CQ" $q$ without constants and having a "variable" which is not part of a join, we obtain $\FMC{q} \polyrx \Shapley{q}$.
Due to space constraints, the adaptation of \Cref{lem:decomposable} is deferred to \Cref{app:endogenous}.
\begin{toappendix}
    We now wish to adapt \Cref{lem:pseudoconn,lem:decomposable} to purely "endogenous" databases. This shall be done by defining a notion of query with an unshared constant, which will depend on the specific case. A "$\aC$-hom-closed" query is said to be ""pseudo-connected with an unshared constant"" if it has an "island" "minimal support" $S$ and a "constant" $a \notin C$ appearing in exactly one "fact" of $S$. It is said to be ""decomposable with an unshared constant"" if it admits a "decomposition" $q_1\land q_2$ that satisfies one of the following:
    \begin{enumerate}
    	\item\label{cas4.1} for $i\in\{1,2\}$, $q_i$ admits a "minimal support" $S_i$ with a "constant" $a_i \notin C_i$ appearing in exactly one "fact" of $S_i$;
    	\item\label{cas4.2} there exists $i\in\{1,2\}$ such that $q_i$ admits a "minimal support" $S_i$ with a "constant" $a_i \notin C_i$ appearing in exactly one "fact" of $S_i$, \textit{and} $\FMC{q_i}\in "FP"$.
    	\\
    \end{enumerate}

	\begin{lemma}[Adaptation of \Cref{lem:pseudoconn,lem:decomposable}]
		Let $q$ be a "$\aC$-hom-closed" query that is either "pseudo-connected with an unshared constant" or "decomposable with an unshared constant". Then $\FMC{q}\polyrx\Shapleyn{q}$.
	\end{lemma}

	Note that this lemma restricted to "pseudo-connected queries with an unshared constant" is a simple reformulation of \Cref{lem:endo-p-con}.
	
	\begin{proof}
		First consider the case where $q$ is "pseudo-connected with an unshared constant" and apply the same construction as with the proof of \Cref{lem:pseudoconn}, by using the "island" "minimal support" $S$ given by the definition and choosing an unshared constant $a$ as the constant to duplicate. Since $q$ is also "pseudo-connected", this construction reduces $\FGMC{q}$ to $\Shapley{q}$. Now by definition $a$ appears in exactly one "fact" of $S$, which means that in the construction $S^0$ is a singleton. This implies that its only element has to be $\mu$ which is "endogenous", hence no "exogenous" "fact" will be added during the construction. In other words, if the input database was purely "endogenous" then the one on which the oracle is called will be as well, hence $\FMC{q}\polyrx\Shapleyn{q}$.\\
		
		When $q$ is "decomposable with an unshared constant", if it is of type (\ref{cas4.1}), we apply the exact same reasoning to compute both $\SPQE{q_i}$ over $\D_i$ separately from oracles to $\Shapleyn{q}$, using the "minimal support" $S_{i'}$ of the other subquery given by the definition. If it is of type (\ref{cas4.2}), we compute $\SPQE{q_{i'}}$ over $\D_{i'}$ this way, using the "minimal support" $S_{i}$ given by the definition, and directly get $\SPQE{q_i}$ on $\D_i$ from the fact that the problem is in "FP". In both cases, we conclude by multiplying the probabilities together.
	\end{proof}
\end{toappendix}



\subsection{Queries with Negation}
\label{sec:negation}
\begin{toappendix}
    \subsection{Queries with Negation}
    \label{app:negation}
\end{toappendix}
$\PQE{}$ and $\Shapley{}$ have been explored beyond "$\aC$-hom-closed" queries, by considering queries with some form of negation. An "FP"/"shP" dichotomy has been established for $\Shapley{}$ over "self-join-free conjunctive queries with safe negations" (sjf-$\CQneg$) in \cite[Theorem 3.1]{reshefImpactNegationComplexity2020}, and another dichotomy for $\PQE{}$ has been shown for the larger class of self-join-free $\intro*\RAneg$ queries \cite[Theorem 1.1]{finkDichotomiesQueriesNegation2016}. In fact, it follows from the proof of \cite{finkDichotomiesQueriesNegation2016} that this latter dichotomy is also valid for $\PQEPhalfOne{}$, since no other probability is used.
\AP
\phantomintro{self-join-free conjunctive queries with safe negations}
Briefly, a sjf-$\intro*\CQneg$ query is a "sjf-CQ" that may have some negative atoms, with the restriction that their variables are present in the positive part. The "hierarchical" queries are defined as before (except the atoms can be negative), and they also characterize the tractable queries.

A slight adaptation of our proof technique yields the following:
\begin{toappendix}
	\subsubsection{Key lemma.}
	
	We start by establishing a key lemma that will allow us to prove the other results about queries with negations.
	\AP A ""DNF formula"" is a disjunction of conjunctions of possibly negated "atoms", which we see as a generalized form of "UCQ".
	
	\begin{lemma}\label{lem:key-neg}
		Let $q= q^+\land \lnot q^- \land \lnot \alpha_1 \land \dots \land \lnot \alpha_K$ be a Boolean query such that:
		\begin{itemize}
			\item $q^+=q^\circ\land q'$ is a "sjf-CQ" whose "constants" are in $\aC$, with $q^{\circ}$ being "variable-connected";
			\item $q^-$ is a "DNF formula" over "constants" of $\aC$ and variables from $\vars(q^+)$, such that every "atom" contains at least one "variable" and every clause contains at least a positive "atom";
			\item $q^-$ shares no "relation name" with $q^+$;
			\item for every $i$, we have that $\alpha_i$ is an "atom" over the "constants" of $\aC$ which shares no "relation name" with $q^+$.
		\end{itemize}
And let $\tilde{q}^-$ be the query resulting from: 
\begin{enumerate}
			\item removing from $q^-$ every negative atom that contains some variable $x\notin\vars(q^{\circ})$ and 
			\item removing every clause whose positive part contains some variable $x\notin\vars(q^{\circ})$.
\end{enumerate}
Further, let $\tilde{q}\defeq q^{\circ}\land \lnot \tilde{q}^-\land \lnot \alpha_1 \land \dots \land \lnot \alpha_K$. Then, $\FGMC{\tilde{q}} \polyrx \Shapley{q}$.
	\end{lemma}
	
	\begin{proof}
		We build from $q^+$ the same reduction as in the proof of \Cref{lem:pseudoconn} (see \Cref{fig:redx}), with $S$ and $S'$ respectively isomorphic to 
		$q$ and $q'$ (see the proof of \Cref{cor:non-h-cq} for details on how this reduction applies to "sjf-CQ"s.)
		. From this we prove the following adaptation of \Cref{lem:pcnx}, with $ \mathcal{F}\defeq \{\mu^{k}\mid k \in [i]\}\cup \{\alpha_k\mid k\in [K]\}$:
		
		\begin{lemma}\label{lem:pcnx.1}
			$(\scorefun[](B\cup\{\mu\}) - \scorefun[](B))=0$ in the following four mutually exclusive cases:
			\begin{enumerate}
				\item\label{cas9.1} $\exists k\in [K] . \alpha_k \in B$,
				\item\label{cas9.2} $\forall k\in [K] .
				\alpha_k \notin B
				~~~\land~~~
				\exists k\in [i] . \mu^{k} \in B$,
				\item\label{cas9.3} $\forall \alpha\in \mathcal{F} .
				 \alpha \notin B
				~~~\land~~~
				\exists \alpha\in  S^- . \alpha \notin B$,
				\item\label{cas9.4} 
				$\forall \alpha\in \mathcal{F} . 
				\alpha \notin B
				~~~\land~~~
				\forall \alpha \in  S^- . \alpha\in B
				~~~\land~~~B\cap \D' \models \tilde{q}$.
			\end{enumerate}
		\end{lemma}
		
		Case \ref{cas9.1} is when there is some $\alpha_k$ that invalidates the query independently of whether $\mu$ is present or not, and the other cases are strictly analogous to those in \Cref{lem:pcnx}. The only difference here is the presence of a negative part.
		The $\alpha_k$ have been dealt with, so let us consider a "minimal support" $M^-$
		of $q^-$ in $B\cup\Aix$. Since $q^+$ and $q^-$ have no "relation name" in common and everything in the reduction other than the input database $\D$ is part of a "minimal support" of $q^+$, then necessarily $M^-\subseteq \D$ (and $M^-\neq\emptyset$ because every clause in $q^-$ contains at least one positive atom). Since every atom in $q^-$ contains a variable, and by construction every variable in $q^+$ is matched to a fresh constant in $S$ or $S'$, these two cannot be affected by $M^-$. Overall, the only impact $M^-$ can have is on a "minimal support" of $q$ in $\D$, but this is taken into account by the condition $B\cap \D' \models \tilde{q}$ in Case \ref{cas9.4}.
	\end{proof}
	
	\subsubsection{Application to sjf-$\CQneg$}
\end{toappendix}
\begin{propositionrep}
	Let $q$ be a sjf-$\CQneg$, where $q^+$ and $q^-$ are its positive and negative atoms, respectively.
    Let $q_{vc}^+$ be any maximal "variable-connected" subquery of $q^+$, and  $q_{vc}^-$ be the atoms of	 $q^-$ whose every variable is in $q^+_{vc}$. Then, $\FGMC{q_{vc}^+\land q_{vc}^-} \polyrx \Shapley{q}$.
\end{propositionrep}
\begin{appendixproof}
	Let $q^- = q^-_{v} \land q^-_{c}$ be, respectively, the atoms of $q^-$ having at least one variable and the atoms having only constants.
	Assume $q^-_c = \set{\lnot \alpha_1, \dotsc, \lnot \alpha_K}$.
    We can now apply \Cref{lem:key-neg} by taking
	\begin{itemize}
		\item $q_{vc}^+$ as $q$;
		\item the DNF of $\lnot q^-_v$ as $q^-$;
		\item the $\alpha_i$'s as above.
	\end{itemize}
	In this context, $\tilde q^-$ will be the "DNF" of \textit{$q_{vc}^-$} with the variable-free atoms removed, hence the result $\FGMC{\tilde{q}} \polyrx \Shapley{q}$ translates to $\FGMC{q_{vc}^+\land q_{vc}^-} \polyrx \Shapley{q}$, as desired.
\end{appendixproof}
The preceding result does not cover the full dichotomy of sjf-$\CQneg$ because of the basic "non-hierarchical" query $A(x)\land \lnot S(x,y) \land B(y)$. However, it does cover all queries with a "variable-connected" positive part,
as well as sjf-$\CQneg$ queries with ``component-guarded negation'',
defined as the sjf-$\CQneg$ queries such that the set of variables of every negative atom appears in the same maximal "variable-connected" subquery of the positive part. 
Finally, our technique also applies to some queries in $\RAneg$ which are not definable in $\CQneg$, for which $\Shapley{}$ was not yet known to be hard  ("cf" \Cref{ex:beyondCQneg:1,ex:beyondCQneg:2} in \Cref{app:1RA-}).

\begin{toappendix}
    \subsubsection{Beyond sjf-$\CQneg$}
    \label{app:1RA-}\textit{}
    
    Thanks to the dichotomy for sjf-\RAneg \cite{finkDichotomiesQueriesNegation2016}, \Cref{lem:key-neg} shows the hardness of $\Shapley{}$ for some queries that cannot be expressed as sjf-$\CQneg$ as evidenced by the two following examples. The notations are directly taken from \cite{finkDichotomiesQueriesNegation2016}, to which we refer the interested reader.
    \begin{example}\label{ex:beyondCQneg:1}
		Consider the "Boolean" $\RAneg$ query 
		\[
			q_1 \defeq \pi_{\emptyset} (D\bowtie (S\bowtie (A-(B-C))))
		\]
		over the "schema" ($D(X)$, $S(X,Y)$, $A(Y)$, $B(Y)$, $C(Y)$), which translates into the following first-order formula:
		\begin{align*}
			q_1 \equiv{}&\exists x,y ~~ D(x)\land S(x,y) \land A(y) \land \lnot (B(y) \land \lnot C(y)) \\
			\equiv {}&\exists x,y ~~ \big(D(x)\land S(x,y) \land A(y) \land \lnot B(y)\big)
    		 \lor{} \\ & \hspace{.82cm} \big(D(x)\land S(x,y) \land A(y) \land C(y)\big)
		\end{align*}
    	The first form illustrates how \Cref{lem:key-neg} can be applied on it by taking $B(y) \land \lnot C(y)$ as $q^-$, while the second shows that this query is not equivalent to a sjf-\CQneg. However, it is a non-"hierarchical" sjf-$\RAneg$ by \cite[Proposition 5.4]{finkDichotomiesQueriesNegation2016} (using pattern $\mathbf P_{6.1}$ from \cite[Figure~9]{finkDichotomiesQueriesNegation2016}) hence its hardness for $\PQEPhalf{}$.
	\end{example}
	\begin{example}\label{ex:beyondCQneg:2}
		Consider now the $\RAneg$ query 
		\[
			q_2 \defeq \pi_{\emptyset}(S - (A \bowtie B))
		\]
		over the schema ($S(X,Y)$, $A(X)$, $B(Y)$), which translates into the following first-order formula:
    	\begin{align*}
    		q_2 \equiv {}&\exists x,y ~~ S(x,y)\land \lnot (A(x)\land B(y))\\
    		\equiv{} &  \exists x,y ~~ \big(S(x,y)\land \lnot A(x)\big)\lor \big(S(x,y)\land \lnot B(x)\big)
    	\end{align*}
    	\noindent Once again \Cref{lem:key-neg} applies, the query is not equivalent to a sjf-\CQneg, and it is a non-"hierarchical" $\RAneg$ (this time using the pattern $\mathbf P_{4.3}$ from \cite[Figure~9]{finkDichotomiesQueriesNegation2016}).
	\end{example}
    
    
    As these examples show, sjf-$\RAneg$ allows for richer negations, that can be nested or contain more than one atom. However, note that, in order for \Cref{lem:key-neg} to apply, these complex negations must always contain some variable in every atom.
    
    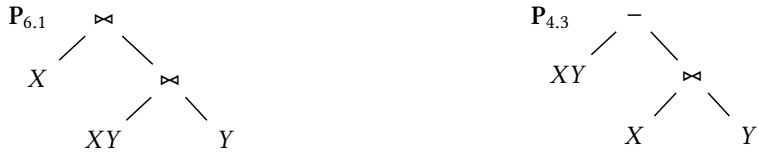
\begin{figure}[ht] 
    	\centering
    	\hfill
$\mathbf P_{6.1}$\hspace{-5mm}
\begin{minipage}[t]{.3\columnwidth}
	\phantom{bla}\vspace{-3mm}\\
	$\begin{tikzcd}[row sep=small, column sep=tiny]
		& \bowtie \arrow[ld, no head] \arrow[rd, no head] &                                                 &   \\
		X &                                                 & \bowtie \arrow[rd, no head] \arrow[ld, no head] &   \\
		& XY                                              &                                                 & Y
	\end{tikzcd}$
\end{minipage}
\hfill
$\mathbf P_{4.3}$\hspace{-5mm}
\begin{minipage}[t]{.3\columnwidth}
	\phantom{bla}\vspace{-3mm}\\
	$\begin{tikzcd}[row sep=small, column sep=tiny]
		& - \arrow[ld, no head] \arrow[rd, no head] &                                                 &   \\
		XY &                                                 & \bowtie \arrow[rd, no head] \arrow[ld, no head] &   \\
		& X                                              &                                                 & Y
	\end{tikzcd}$
\end{minipage}
\hfill     	
    	\caption{%
    		Reproductions of patterns from \cite[Figure~9]{finkDichotomiesQueriesNegation2016}.
    	}
    \end{figure}
\end{toappendix}

\subsection{Maximum Shapley Value}
\label{sec:maximum}
A main reason for considering the Shapley value is to identify the most important facts for a given query. 
One might therefore consider 
the problem of computing the top contributor and its "Shapley value". 
\AP
Concretely, let $\intro*\maxShapley q$ be the problem of, given a "database" $\D$, outputting any "fact" $\alpha \in \Dn$ and its "Shapley value" $v_\alpha$ for~$q$, such that $v_\alpha$ is the maximum value among all ("endogenous") "facts".
It turns out that all of our reductions to show $\FGMC q \polyrx \Shapley q$ can be adapted to show $\FGMC q \polyrx \maxShapley q$, leading to the conjecture that $\maxShapley{}$ might be equivalent to $\Shapley{}$.

\begin{proposition}\label{prop:maxshp}
    For any pair of queries $q,q'$ for which we obtain a reduction $\FGMC q \polyrx \Shapley{q'}$ as a consequence of \Cref{lem:pseudoconn}, \ref{lem:leak} or \ref{lem:decomposable}, we have  $\FGMC q \polyrx \maxShapley{q'}$.
\end{proposition}

The proof of Proposition \ref{prop:maxshp} utilizes the following lemma that shows that the Shapley value of a fact that is a "generalized support" on its own is always maximal. This property holds for the class of games that are monotone ($B\subseteq B'\Rightarrow \scorefun[](B)\le\scorefun[](B')$) and binary ($\scorefun[]$ has  image $\set{0,1}$).

    \begin{lemma}\label{lem:singleton-gz-sup}
    	Let $(P,\scorefun[])$ be a monotone 
    	    	binary 
    	Shapley "game" that contains some $s\in P$ such that $\scorefun[](\{s\})=1$. Then for every $p \in P$,  $\Sh(P,\scorefun[],p)\le\Sh(P,\scorefun[],s)$.
    \end{lemma}

\begin{proof}
For every $p\in P$, denote $\W_p \defeq \{\sigma \in \Sym(P) \mid  \scorefun[](\sigma_{<p} \cup \{p\}) - \scorefun[](\sigma_{<p}) = 1\}$. Since the "game" is monotone and binary, $\scorefun[](\emptyset)=0$ and $\scorefun{\{s\}}=1$, the $\W_p$ form a partition of $\Sym(P)$, and by \Cref{def_sh} $\Sh(P,\scorefun[],p)$ is proportional to $|\W_p|$.

Now let $p\in P$ and $\sigma\in\W_p$. Since the "game" is binary, then necessarily we must have $\scorefun[](\sigma_{<p}) = 0$, which implies $s\not\in\sigma_{<p}$ because the "game" is monotone. Therefore, if $\tau$ denotes the function that swaps the positions of $s$ and $p$ in a permutation, then $\scorefun[](\tau(\sigma)_{<s} \cup \{s\}) - \scorefun[](\tau(\sigma)_{<s}) = 1$ because $\tau(\sigma)_{<s}\subseteq\sigma_{<p}$ and $\scorefun[](\sigma_{<p}) = 0$. This means that $\tau_{|\W_p}$ injects in $\W_s$ hence $|\W_p|\le|\W_s|$. It follows that $\Sh(P,\scorefun[],p)\le\Sh(P,\scorefun[],s)$.
%
\end{proof}

\begin{proof}[Proof of Proposition \ref{prop:maxshp}]
    %
%
	Note that if we assume that $\Dx \nvDash q$ (otherwise every fact has a zero hence maximal Shapley value), a fact $\mu$ that is a "generalized support" on its own indeed verifies $\scorefun[](\{s\})=1$. Now if we go back to the constructions used to prove \Cref{lem:pseudoconn,lem:leak,lem:decomposable}, they are equally valid if we take $S^0\defeq S$ and $S^-\defeq\emptyset$ (the only point of $S^-$ was to limit the number of "exogenous" facts). If we do so, then $\mu$, which is the only fact the $\Shapley{q}$ oracle is called on, is a singleton "generalized support" and therefore has a maximal Shapley value.
\end{proof}
%

Observe that $\maxShapley{}$ is not the same problem as the ``maximum $\Shapley q$ contributor problem'', that is, the problem of finding a "fact" with maximum "Shapley value". We believe this is a relevant problem which may, in principle, be simpler to address than $\Shapley{}$.

\subsection{Shapley Value of Constants}
\label{sec:Shapley-constants}
Instead of computing the "Shapley values" of "facts", one could consider a setting in which the players are a set of ``endogenous constants''. Likewise, one could consider supports consisting of constants, or assign independent probabilities to the "constants"  instead of "facts".

As an example, consider a simple "schema" containing two binary relations {\sf Publication(authorID, paperID)} and {\sf Keyword(paperID, keywordStr)} containing publications and related keywords, and the query $q^* = \exists x,y ~ \textsf{Publication}(x,y) \land \textsf{Keyword}(y,\texttt{`Shapley'})$ testing the existence of \texttt{`Shapley'}-related papers. 
One could then 
extract the experience level of authors on this topic by computing the "Shapley value" of author "constants",
treating all "constants" as "exogenous" except those in the \textsf{authorID} column. Note that the Shapley value for facts would be much less informative as an author's expertise may be split across several facts.

This idea gives rise to variants of $\PQE{}$, $\FGMC{}$ and $\Shapley{}$. 
\AP
For any "database" $\D$ and $C \subseteq \Const$, let $\D|_C$ denote the "database" induced by the "constants" $C$, "ie",\siderev{Explicit game definition} \changed{$\D|_C \defeq \set{\alpha \in\D \mid \const(\alpha)\subseteq C}$.}
For a monotone query $q$, 
let $\intro*\ShapleyConst{q}$ be the task of computing, for 
a "database" $\D$, 
a partition $\const(\D) = C_{\subendo} \dcup C_{\subexo}$, and 
a "constant" $c \in C_{\subendo}$, 
the "Shapley value" in the "cooperative game" where the set of players is $C_{\subendo}$
and the "wealth function" assigns $1$ to any subset $C \subseteq C_{\subendo}$ such that $\D|_{C \cup C_{\subexo}} \models q$ and $\D|_{C_{\subexo}} \not\models q$, or $0$ otherwise.
\AP
Similarly, let $\intro*\FGMCConst{q}$ be the task of computing, for any $k \in \Nat$, "database" $\D$ and $\const(\D)=C_{\subendo} \dcup C_{\subexo}$, the number of sets $C \subseteq C_{\subendo}$ of size $k$ such that $\D|_{C \cup C_{\subexo}} \models q$.
\AP The variants $\intro*\ShapleynConst{q}$ and $\intro*\FMCConst{q}$ (and the probabilistic variants) can be defined analogously.


By a simple adaptation of our proof of \Cref{lem:pseudoconn}, we can show a polynomial equivalence result for these variants: 

\begin{proposition}\label{prop-sh-constants}
    For every "hom-closed" query $q$, we have
    $\ShapleyConst{q} \polyeq \FGMCConst{q}$ and $\ShapleynConst{q} \polyeq \FMCConst{q}$.
\end{proposition}
\begin{proof}
    The reduction from $\ShapleyConst{q}$ (resp.\ $\ShapleynConst{q}$) to $\FGMCConst{q}$ (resp.\ $\FMCConst{q}$) can be obtained by directly applying the proofs of \Cref{prop:reductions-problems} and \Cref{cor:shn-tract}. For the other direction, we shall prove the stronger result that for any "$C$-hom-closed" query $q$, $\FGMCConst{q}(\D)$ (resp. $\FMCConst{q}(\D)$) can be computed in polynomial time from an oracle to $\ShapleyConst{q}$ (resp. $\ShapleynConst{q}$) provided that $\D\cap C\subseteq\Dx$. The argument is a simple adaptation of the proof of \Cref{lem:pseudoconn}.
    
    Let $\D$ be an input database such that $\D\cap C\subseteq\Dx$. If all "minimal supports" of $q$ have their constants in $C$, then either all coalitions are "generalized supports" or none are, but either way the computation of $\FGMCConst{q}(\D)$ is trivial. Otherwise, take a "minimal support" $S$ of $q$ such that $\const(S)\nsubseteq C$, and by $C$-"homomorphically" renaming it if necessary we can assume $\const(S)\setminus C$ is reduced to a singleton $\{a_{\mu}\}$ with $a_{\mu}\notin\const(\D)$.
    Since the players are the constants, $S$ is essentially a "duplicable singleton support" (see \Cref{cor:singleton-sup}): in the formula that defines $\Sh(C_{\subendo} \cup \{a_{\mu}\}, \scorefun, a_{\mu})$, any coalition that contains $a_{\mu}$ will satisfy the query, and any that does not will not be contained in the input database.
    
    Finally, note that no "exogenous" element ("constants" instead of "facts" in this setting) is added since $S$, the minimal support used in the construction, only contains a single "constant".
\end{proof}
\noindent In fact, this proposition can be extended to apply more generally to "$C$-hom-closed" queries provided the constants in $C$ are "exogenous" -- in particular, capturing the example query $q^*$. 

We
note that in the "graph database" setting these variants are equivalent to considering "Shapley values" of \emph{nodes} instead of \emph{edges}. This problem was considered in \cite[\S 6]{khalilComplexityShapleyValue2023}, where the authors showed that $\ShapleyConst{q}$ is "shP"-hard for any full "CRPQ" with a ``non-redundant'' atom whose language contains a word of length at least four.

\section{Discussion}
\label{sec:discussion}
After having identified $\FGMC{}$ as the suitable counting analog of $\Shapley{}$, 
we have exhibited reductions from $\FGMC{}$ to $\Shapley{}$.  
As the hypotheses of our reductions are formulated generally, they
can be combined with existing hardness results for model counting or probabilistic query evaluation 
to establish "shP"-hardness for $\Shapley{}$
for a range of query classes. In particular, we have obtained an "FP"/"shP" dichotomy
for constant-free unions of connected CQs, a relevant stepping stone towards 
the conjectured dichotomy for UCQs. 
In some cases, our reductions 
allow us not only to transfer hardness results but also
to establish the polynomial-time equivalence between $\FGMC{q}$ 
and $\Shapley{q}$, providing the strongest evidence as of yet that Shapley value computation is, from a complexity perspective, nothing other 
than a counting problem. 


%

As is apparent from Figure \ref{fig:query-lang}, most of our new results concern constant-free queries, as  
the presence of constants in queries make reductions from $\FGMC{}$ considerably more difficult. While there are known methods for eliminating constants from queries in the probabilistic setting (most notably, the so-called query ``shattering'' \cite[\S 2.5]{dalviDichotomyProbabilisticInference2012}), they interact badly 
with our hypotheses. For example, the shattering of a connected
 "variable-connected" query may become disconnected, and not even "pseudo-connected", as illustrated by \Cref{ex:bad-shattering}.
It is therefore remains an important question for future work how best to eliminate or otherwise handle constants when reducing $\FGMC{}$ to $\Shapley{}$.
\begin{toappendix}
\begin{example}\label{ex:bad-shattering}
	Consider the $CQ$ $R(x,y)\land S(a,x)\land S(x,a) \land T(x,z)$. It is "variable-connected" since every atom contains "x" but its complete shattering \cite{dalviDichotomyProbabilisticInference2012} contains the disjunct $R_{a,*}(y)\land S_{a,a}()\land T_{a,*}(z)$ which isn't "connected".
\end{example}
\end{toappendix}

Aside from extending our reductions to larger classes of queries, there are other interesting questions 
to explore related to the variants from Section \ref{sec:extensions}. 
In particular, we would like to know whether the purely endogenous variant $\Shapleyn{}$ is as hard as $\Shapley{}$, which appears to be related to the open question of whether there exist queries for which $\MC{}$ is tractable but not $\GMC{}$. 
We are also eager to understand better the Shapley value of constants, which we find natural. Due to \Cref{prop-sh-constants}, 
we can alternatively consider the model counting analog $\FGMCConst{q}$,
which is conceptually simpler.

 
\begin{acks}
  Diego Figueira is partially supported by ANR QUID, grant ANR-18-CE40-0031.
  Meghyn Bienvenu and Diego Figueira are partially supported by ANR AI Chair INTENDED, grant ANR-19-CHIA-0014.
  Part of this research was produced while Pierre Lafourcade was supported by the ENS Paris-Saclay.
\end{acks}

\bibliographystyle{ACM-Reference-Format}

\end{document}